\documentclass[review]{elsarticle}

\usepackage{lineno,hyperref}
\usepackage{multirow}
\usepackage{graphicx}
\usepackage{enumitem}

\modulolinenumbers[5]

\journal{Journal of \LaTeX\ Templates}

\usepackage{booktabs}
\usepackage{amsthm}
\usepackage{amsmath}

\renewcommand{\paragraph}[1]{\vspace{1mm}\noindent{\bf #1}}
\newcommand{\eat}[1]{}

\newtheorem{definition}{Definition}[section]
\newtheorem{proposition}[definition]{Proposition}
\newtheorem{lemma}[definition]{Lemma}

\newtheorem{theorem}[definition]{Theorem}

\newtheorem{example}[definition]{Example}

\newcommand{\squishlist}{
\begin{list}{$\bullet$}
{ \setlength{\itemsep}{0pt} \setlength{\parsep}{3pt}
\setlength{\topsep}{3pt} \setlength{\partopsep}{0pt}
\setlength{\leftmargin}{1.5em} \setlength{\labelwidth}{1em}
\setlength{\labelsep}{0.5em} } }

\newcommand{\squishlisttwo}{
\begin{list}{$\bullet$}
{ \setlength{\itemsep}{0pt} \setlength{\parsep}{0pt}
\setlength{\topsep}{0pt} \setlength{\partopsep}{0pt}
\setlength{\leftmargin}{2em} \setlength{\labelwidth}{1.5em}
\setlength{\labelsep}{0.5em} } }

\newcommand{\squishend}{
\end{list}  }

\newcommand{\symif}{:-}

\newcommand{\CM}{{\cal M}}
\newcommand{\C}{{\cal C}}
\newcommand{\V}{{\cal V}}
\newcommand{\R}{{\cal R}}
\newcommand{\I}{{\cal I}}
\newcommand{\CP}{{\cal P}}
\newcommand{\CS}{{\cal S}}

\newcommand{\CL}{{\cal L}}
\newcommand{\CQ}{{\cal Q}}
\newcommand{\CA}{{\cal A}}
\newcommand{\CB}{{\cal B}}


\begin{document}

\begin{frontmatter}

\title{Semi-interval Comparison Constraints in Query Containment and Their Impact on Certain Answer Computation}
%
\author{Foto N. Afrati}
\ead{afrati@gmail.com}
\address{National Technical University of Athens}
\author{Matthew Damigos}
\ead{mgdamig@gmail.com}
\address{Ionian University, Corfu}

\begin{abstract}

We consider conjunctive queries with arithmetic comparisons (CQAC) and investigate the computational complexity of the problem: Given two CQAC queries, $Q$ and $Q'$, is $Q'$ contained in $Q$? We know that, for CQAC queries, the problem of testing containment is $\Pi_2 ^p$ -complete.    However, there are broad classes of queries with semi-interval arithmetic comparisons in the containing query that render the problem solvable in NP. In all cases examined the contained query is allowed to be any CQAC. Interestingly, we also prove that there are simple cases where the problem remains $\Pi_2 ^p$ -complete.

 We also investigate the complexity of computing certain answers in the framework of answering CQAC queries with semi-interval comparisons using any CQAC views. We prove that maximally contained rewritings in the language of union of CQACs always compute exactly all certain answers. We find cases where we can compute certain answers in polynomial time using maximally contained rewritings.

\end{abstract}

\begin{keyword}

\end{keyword}

\end{frontmatter}


\section{Introduction}
\label{intro-sect}

A conjunctive query with arithmetic comparisons (CQAC) is a select-project-join query in SQL. Query containment and query equivalence play a prominent role in efficient query processing, e.g., for minimizing the number of joins in a query. Query equivalence can be reduced to a query containment problem. Data integration is often put into the framework of answering queries using views via contained rewritings that are found based on properties of query containment. Recently, the problem of determinacy has attracted the interest of researchers and query containment tests offer tools for its investigation.

For conjunctive queries (CQs), the query containment problem is shown to be NP-complete~\cite{Chandra77-1}. 
Membership in NP is proven via a containment mapping from the variables of one query to the variables of the other which preserves the  relation symbol, i.e., it is  a homomorphism. 
 For conjunctive queries with arithmetic comparisons the query containment problem is $\Pi^p_2$-complete~\cite{Klug88-1},
\cite{Meyden97-12}.  

We denote conjunctive queries with arithmetic comparisons (CQAC) by $Q=Q_0+\beta$ where $Q_0$ denotes the relational subgoals and $\beta$ the arithmetic comparison subgoals.
The containment test now uses  all containment mappings $\mu_1,\dots , \mu_k$ from the variables of one query to the variables of the other.
The containment test decides whether $Q_2\sqsubseteq Q_1$ by checking whether the following containment entailment is true~\cite{Gupta94-1,Zhang94-4}:
$$\phi: \beta_2  \Rightarrow \mu_1(\beta_1) \vee \cdots \vee
	\mu_k(\beta_1).\footnote{we assume the queries are normalized; see definition shortly.}$$
Previous work has considered CQACs with only left semi-interval arithmetic comparisons or only right semi-interval ACs~\cite{Klug88-1,AfratiLM06}.  A left semi-interval  (LSI) AC is an AC of type $var\leq const$ or $var< const$, where $var$ is a variable and $const$ is a constant and a right semi-interval  (RSI) AC is an AC of type $var\geq const$ or $var> const$.   The LSI AC with $\leq$ is called closed (CLSI) and the LSI AC with $<$ is called open
(OLSI) and similarly for RSI, we define ORSI and CRSI.
Klug in~\cite{Klug88-1} noticed that if only LSI (RSI respectively) are used then one containment mapping suffices to prove query containment. Further work  in~\cite{Afrati04-1} considers certain broader classes of CQACs with LSI (RSI respectively) and shows that a single containment mapping suffices to prove query containment (in this case, we say that the homomorphism property holds). In~\cite{AfratiLM06} a more elaborate approach that works on the containment entailment is taken to prove that for certain cases of queries with both LSI and RSI ACs, we can check containment in non-deterministic polynomial time although we may need more than one containment mappings to make the containment entailment true. 

\begin{table*}
  \small{ \begin{center} \begin{tabular} {|c|c|l|c|c|} \hline
{\bf Contained Query} & {\bf Containing Query} &  {\bf norm} &{\bf Complexity} & {\bf
Reference} \\\hline \hline

 OSI and $\neq$&   OSI   &n/a &     $\Pi^p_2$-complete        &  Theorem \ref{thm-hardness}  \\\hline
 LSI & OLSI, constant    &n/a &     $\Pi^p_2$-complete        &  Theorem \ref{thm-hardness123}  \\\hline
 CLSI, $\neq$  &  OLSI   &n/a &     $\Pi^p_2$-complete        &  Theorem \ref{thm-hardness123}  \\\hline
  $\neq$  &  $\neq$   &n/a &     $\Pi^p_2$-complete        &  Theorem \ref{thm-hardness123}, \cite{Kolaitis98-1}  \\\hline
%

    any & one AC & no & NP & Theorem \ref{thm-one} \\\hline

  CLSI  &   CLSI      & no   &   NP/HP          & \cite{Klug88-1} \\  \hline

  any  &   CLSI      & no   &   NP/HP          & Theorem \ref{them-norm-lsi}, \cite{Afrati04-1} \\  \hline

     any$^c$ &LSI  & no & NP/HP &  Theorem \ref{them-norm-lsi}, \cite{Afrati04-1}\\ \hline 

SI &   CLSI 1CRSI   & yes&        NP     &  \cite{AfratiLM06} \\\hline

any closed ACs& CLSI 1CRSI     & no&     NP        &  \cite{Afrati19} \\\hline

 any& CLSI 1CRSI     & no&     NP        &  Theorem \ref{them-norm-lsi-rsi1} \\\hline

  any$^c$& CLSI 1ORSI     & no&     NP        &  Theorem \ref{them-norm-lsi-rsi1} \\\hline

  any$^c$& OLSI 1ORSI     & no&     NP        & Theorem \ref{them-norm-lsi-rsi1} \\\hline

  any$^c$& OLSI 1CRSI     & no&     NP        &  Theorem \ref{them-norm-lsi-rsi1} \\\hline




%
\end{tabular}
\caption{Complexity of CQAC query containment focusing on cases the containing query uses only semi-interval comparisons.
}

\label{fig:table-results}
\end{center} }  \end{table*}

In data integration, in the local-as-view approach, views are maintained over the sources in order to be used to answer queries. Query answering is usually made possible via a rewriting that expresses the query in terms of the views. 
Since the views do not provide all the information that the base relations that form the query would require, we are looking into computing certain answers, i.e., all the answers that we are certain would be in the answers of the query given a specific view instance. 
For conjunctive queries, it is shown that maximally contained rewritings  (MCR) in the language of union of conjunctive queries compute all certain answers in polynomial time \cite{AfratiK10}. For CQAC queries and views, this may not  always be the case. There is a large body of work on the topic of obtaining MCRs for certain cases \cite{2017Afrati}. Recent work developed an implementation for computing certain answers~\cite{BenediktCGGTMO22}.  In \cite{Abiteboul98-1}, the problem of finding certain answers is proven coNP-hard when the queries have inequalities ($\neq$) and the views are CQs.  

When only CLSI are used in the containing query the containment problem is in NP via the homomorphism property.
Interestingly, for the case the containing query uses only OLSIs, we prove here that the containment problem is still 
$\Pi^p_2$-complete. However, to prove the hardness result we need to use queries which share constants in a rather  intricate way, thus creating a dichotomy which  leaves a class with the ``majority'' of the queries to be such that query containment is an NP problem.  In fact, the subcase for which complexity is in NP,  membership in NP is shown via the homomorphism property, i.e., one mapping suffices to prove containment. 
The only other result we are aware of concerns CQAC queries that use only $\neq$ comparisons
 \cite{Kolaitis98-1} where a boundary is delineated depending on the number of relational subgoals of the query.

The following example shows two queries that belong in the class for which containment checking is not an NP problem:


$Q_1() :- a(X,5),  X< 5$~~~~~~~$Q_2() :- a(X,5), a(Y,X),X\leq 5, Y<5$

The above example shows another challenge that we need to take into account  before we find all the containment mappings to use in the containment entailment: we need to normalize the queries. A CQAC query is normalized if each variable appears only once in relational subgoals and to compensate for that, we add equalities.
Also no constants are allowed  in relational subgoals. Thus, a normalized query uses explicit equalities in addition to inequalities. The above queries in the normalized form are as follows:\\
$Q_1() :- a(X,Z),  X< 5, Z=5$\\
$Q_2() :- a(X,Z), a(Y,W), X\leq 5, Y<5,Z=5,W=X$

The cases that we treat in this paper and all related previous results are listed in 
Table \ref{fig:table-results}.  where, in the first two columns we mention the arithmetic comparisons allowed.
Notation  any$^c$  in the table means that any AC can be used in $Q_2$ up to a certain condition which is stated in the corresponding theorem.  Notation NP/HP means that it is in the complexity class NP via the homomorphism property. The third column in the table refers to whether normalization is needed.   
The table only contains results from the literature that concern SI comparisons. All the symmetrical results are valid  too,  i.e., when we  interchange  RSI with LSI ACs and vice versa. We do not state explicitly the symmetrical results.

 In the framework of answering queries using views,  
when only a view instance is available, we want to find answers to a given query which are always computed on any database that produces the view instance or a superset of the view instance. Through this consideration the concept of certain answers is defined. 
Query rewriting  techniques are used to find efficiently some certain answers to queries when 
 only a view instance is available. Actually, for CQ query and views, maximally contained rewritings  (MCR) can compute exactly all certain answers, and hence certain answers can be computed in polynomial time.
In this paper, we prove that the same result is true for CQAC query and views. 
However, unlike CQs, it is not easy to find MCRs that are unions of CQAC queries and the problem of finding certain answers for CQ views and queries with only ``$\neq$'' is coNP-hard, as is proven in  \cite{Abiteboul98-1}. Here we focus again on queries with LSI and RSI ACs only and show that, for any CQAC view set, we can find an MCR in the language of (possibly infinite) union of CQACs.
This extends the results in  \cite{AfratiLM06}  and settles an error in there because, as the following example demonstrates, even when the views use only SI ACs in their definition, an MCR may have to use  an AC that relates two variables. The following query and views use only SI ACs: \\ $Q() \symif e(X,Y),e(Y,Z),X\geq 5,Z\leq 5$, \\ $V_1(Z) \symif e(X,Y),e(Y,Z),X\geq 5$,   $V_2(X) \symif e(X,Y),e(Y,Z),Z\leq 5$.\\
 However, the  following  CQAC:  
$Q() \symif  V_1(Z),V_2(X),X\geq Z$ is a contained rewriting that uses an AC that is not an SI.
 The cases that we treat in this paper and all related previous results are listed in 
Table \ref{fig:table-results2}.  
\begin{table*}
  \small{ \begin{center} \begin{tabular} {|c|c|l|c|c|} \hline
{\bf Views  } & {\bf Query } &  {\bf MCR} &{\bf Certain } & {\bf
Reference} \\
{\bf } & {\bf } &   &{\bf Answers} &  \\
\hline \hline

CQ-LSI&   CQ-LSI  &Union of CQ-LSI &     PTIME       &  \cite{Pottinger00-1,AfratiLM06} \\\hline

CQ-SI&   CQ-RSI1  &Datalog with ACs &     PTIME       &  \cite{AfratiLM06} \\\hline
CQAC&   CQ with closed RSI1  &Datalog with ACs &     PTIME       & Thm. \ref{thm-mainsec61certain}
 \\\hline
 CQAC&   CQ$^*$ w. closed/open RSI1  &Datalog with ACs &     PTIME       &  Thm. \ref{thm-mainsec61certainMCR}  \\\hline
CQAC&   CQ$^*$ w. open/closed RSI1  &Datalog with ACs &     PTIME       & Thm.  \ref{thm-mainsec61certainMCR} \\\hline

CQAC&   CQ$^*$ w. open RSI1  &Datalog with ACs &     PTIME       &  Thm. \ref{thm-mainsec61certainMCR} \\\hline
CQ&   CQ$^{\neq}$   & ~~~~~~~~n/a  &  coNP-hard      &   \cite{Abiteboul98-1} \\\hline



\end{tabular}
\caption{Work on finding maximally contained rewritings (“MCR”) and certain answers.}

\label{fig:table-results2}
\end{center} }  \end{table*}


The main contributions in this paper are:
\squishlist
\item The results on the complexity of the query containment problem for conjunctive queries with arithmetic comparisons  (CQACs) where the containing query uses only semi-interval arithmetic comparisons. These results are summarized in
Table
\ref{fig:table-results} where we mention only results from previous work that concern semi-interval comparisons.
 (Sections \ref{sec-hardd} -- \ref{sec-exte-single-mappinga-var})
\item We prove that for  CQAC queries and views, if there is a maximally contained rewriting (MCR) in the language of union of CQACs, then this MCR computes exactly all certain answers. (Sections \ref{sec-mcr-intro} -- \ref{sec-6} ) 
\item  For queries that are CQAC with semi-interval comparisons with a single RSI comparison  and  any CQAC views, we build an MCR in the language of Datalog with arithmetic comparisons. (Sections  \ref{sec-mcrs-cert-answe} -- \ref{subsec-mcr-datalogAC}), hence proving that in this case, we can compute certain answers in polynomial time.
\squishend

The structure of the paper is as follows:
After Related Work (Section   \ref{sec-related-}) and Preliminaries (Section \ref{sec-prelims}),  Section \ref{sec-reason-AC}  presents a sound and complete set of elemental implications to derive an arithmetic comparison from a given set of arithmetic comparisons. It also presents the preliminaries to analyze the containment test in the case of CQACs.
Section \ref{sec-hardd} presents the reduction for the hardness result.

The three sections that follow investigate the cases where the problem of query containment is in NP. Section   
\ref{sec-one-ac-cont} considers the case where the containing query has only one AC.
 Section  \ref{sec-si-no-norm}   considers the case where the  containing query uses SI ACs. Section \ref{sec-normaliz}  discusses the issue of normalization and extends the results of Section  \ref{sec-si-no-norm}. Section  \ref{sec-exte-single-mappinga-var}  takes advantage of the observation that the head variables (and possibly some more variables) of the containing query map on the same variables of the contained query for every containment mapping and, thus, extends the results of the previous sections.
%
%
%
%

Section \ref{sec-mcr-intro} introduces maximally contained rewritings (MCRs) and Section \ref{sec-mcr-more} defines expansion of a rewriting and makes remarks that concern idiosyncrasy for CQAC and contained rewritings. Section \ref{sec-6} proves the next major result which says that an MCR in the language of union of CQACs computes all certain answers for CQAC query and views. The rest of the sections focus on finding MCRs for the case the query contains LSI ACs and only one RSI AC. For this to happen, we make use of another containment test, for this particular case, which is based in a transformation of the containing query to a Datalog query (without ACs) and a transformation of the contained query to a CQ. We present this test in Sections \ref{sec-mcrs-cert-answe} and \ref{subsec-construct-Datalog} which leads to Theorem \ref{thm:main123}, the proof of which is in  \ref{prf:thm-main123p} while  Section \ref{subsec-simplefacts} serves as introduction to the proof in  \ref{prf:thm-main123p}. In Section \ref{subsec-mcr-datalogAC} the algorithm based on Theorem \ref{thm:main123} is presented for finding an MCR for the aforementioned special case.


\section{Related work }

\label{sec-related-}

\textit{CQ and CQAC containment:} The problem of containment between conjunctive queries (CQs, for short) has been studied in \cite{Chandra77-1}, where the authors show that the problem is NP-complete, and the containment can be tested by finding a containment mapping. As we already mentioned, considering CQs with arithmetic comparisons (CQACs), the problem of query containment is $\Pi^p_2$-complete \cite{Meyden97-12}. 
Zhang and Ozsoyoglu, in \cite{Zhang93-1}, showed that  testing containment of two CQACs can be done by checking the containment entailment.
Kolaitis et al. \cite{Kolaitis98-1} studied the computational complexity of the query-containment problem of queries with disequations ($\neq$). In particular, the authors showed that the problem remains $\Pi^p_2$-hard even in the cases where the acyclicity property holds and each predicate occurs at most three times. However, they proved that if each predicate occurs at most twice then the problem is in coNP. 
Karvounarakis and Tannen, in \cite{karvounarakis2008conjunctive}, also studied CQs with disequations ($\neq$) and identified special cases where query containment can be tested by checking for a containment mapping (i.e., the containment problem for these cases is NP-complete).

The homomorphism property for query containment of conjunctive queries with arithmetic comparisons was  studied in \cite{Klug88-1,Zhang94-4,Afrati04-1,Afrati19}, where classes of queries were
%
identified for which the homomorphism property holds.

\vspace{10px}
\textit{Rewritings and finding MCRs:} The problem of answering queries using views has been extensively investigated in the relevant literature (e.g., \cite{Levy95-1, Levy00-1, 2017Afrati}); including finding equivalent and contained rewriting. Algorithms for finding maximally contained rewritings (MCRs) have also been studied in the past \cite{Abiteboul98-1, Grahne99-1, Levy96-1, Pottinger00-1, Mitra99-1, Duschka97-3, Afrati99-1, AfratiLM06}.  The authors in \cite{Pottinger00-1} and \cite{Mitra99-1} propose two algorithms, the Minicon and shared-variable algorithm, respectively, for finding MCRs in the language of unions of CQs when both queries and views are CQs.  \cite{Pottinger00-1} also considers restricted cases of arithmetic comparisons (LSI and RSIs) in both queries and views.
\cite{GeckKSS23} examines equivalent  CQ rewritings for acyclic CQ queries.      
The works in \cite{Duschka97-3} and \cite{Afrati99-1} studied the problem where the query is given by a Datalog query, while the views are given by CQs and union of CQs, respectively. In both papers, the language of MCRs is  Datalog. The authors in \cite{CaoFGL18} studied the problem of finding MCRs in the framework of bounded query rewriting. They investigated several query classes, such as CQs, union of CQs, and first order queries, and analyzed the complexity in each class. Work in \cite{AfratiK10} proposed an efficient algorithm that finds MCRs in the language of union of CQs in the presence of dependencies. The work in \cite{AfratiLM06} investigated the problem of finding MCRs for special cases of CQACs. \cite{CimaCLP23} is a recent account on view-based query processing as  an abstraction in data integration. Determinacy is another related problem investigated recently, where we ask about  the existence of equivalent rewritings in the case where the answers to the views uniquely determine
the answers to the query. In 
\cite{NashSV07}, \cite{SegoufinV05}, \cite{Afrati-determ}, \cite{CalvaneseGLV02}, \cite{CalvaneseGLV05} notions related to determinacy are considered, in \cite{BenediktPW23}    determinacy for nested relational queries is investigated, in \cite{KwiecienMO22}  determinacy for multiset semantics, and \cite{BenediktEM17}       \cite{Marcinkowski20}, \cite{BenediktKOR23} investigates determinacy for recursive queries and views.

\vspace{10px}
\textit{Certain answers and MCRs:} The problem of finding certain answers has been extensively investigated in the context of data integration and data exchange, the last 20 years (e.g., \cite{AndritsosFFHHHKMNPVVY02, Abiteboul98-1, Grahne99-1, fagin2005data, AfratiK10, KonstantinidisA13}). In \cite{Grahne99-1, fagin2005data}, the authors investigated the problem of finding certain answers in the context of data exchange, considering CQs. The work in \cite{fagin2005data} was extended for arithmetic and linear arithmetic CQs  in \cite{CateKO13}.
In \cite{Abiteboul98-1}, the authors investigated the relationship between MCRs and certain answers.
 In \cite{Abiteboul98-1}, the problem of finding certain answers is proven coNP-hard when the queries have inequalities ($\neq$) and the views are CQs. 
In \cite{AfratiK10}, the authors proved that an MCR of a union CQs computes all the certain answers, where MCR is considered in the language of union of CQs.  \cite{BenediktCGGTMO22} developed an implementation for computing certain answers.

\vspace{10px}
\textit{Other work with arithmetic comparisons in queries:} As concerns studying other related problems of queries  in the presence of arithmetic comparisons recent work can be found
in \cite{FanLLT18}, where the authors propose to extend
graph functional dependencies with linear arithmetic expressions
and arithmetic comparisons. They study the problems of testing satisfiability and related problems over integers (i.e., for non-dense orders).
In  \cite{PapadimitriouY97} the complexity  of evaluating conjunctive queries  with
arithmetic comparisons is investigated  for acyclic queries, while query containment for acyclic conjunctive queries was investigated in \cite{ChekuriRaj97}.
Other works \cite{CateKO13,AfratiLP08}  have added arithmetic to extend the expressiveness of  tuple generating dependencies and data exchange mappings, and studied the complexity of related problems. 
Queries with arithmetic comparisons on incomplete databases are considered in \cite{ConsoleHL20}.

\section{Preliminaries}
\label{sec-prelims}


A \textit{relation schema} is a named relation defined by its name (called \textit{relational symbol}) and a vector of attributes. An \textit{instance} of a relation schema is a collection of tuples with values over its attribute set. These tuples are called {\em facts}. The schemas of the relations in a database constitute its \textit{database schema}. A relational \textit{database instance} (database, for short) is a collection of  relation instances.

A {\em conjunctive query (CQ for short)} $Q$ over a database schema $\CS$ is a query of
the form: $h( \overline{X})\ :-\  e_1( \overline{X}_1),\ldots, e_k( \overline{X}_k)$,
where $h( \overline{X})$ and $ e_i( \overline{X}_i)$ are atoms, i.e., 
they contain a relational symbol (also called \textit{predicate} - here, $h$ and $e_i$ are predicates) and a
vector of variables and constants.
The atoms that contain only constants are called \textit{ground} atoms and they represent   \textit{facts}.

The {\em head} $h(\overline{X})$, denoted $head(Q)$,
represents the results of the query, and $e_1 \ldots e_k$ represent
database relations (also called base relations) in $\CS$.
The variables in $\overline{X}$
are called \textit{head} or \textit{distinguished} variables,
while the variables in $\overline{X}_i$ and not in  $\overline{X}$  are called \textit{body} or \textit{nondistinguished}
variables of the query. 
The part of the conjunctive query on the right of symbol $:-$ is called the {\em body} of the query and is denoted $body(Q)$.
Each atom in the body of a conjunctive query
is said to be a {\em subgoal}. 
A conjunctive query is said to be
\textit{safe} if all its distinguished variables also occur in its
body. We only consider safe queries here.

The \textit{result} (or \textit{answer}), denoted $Q(D)$, of a CQ $Q$ when it is applied on a database instance $D$ is the set of atoms such that for each assignment $h$ of variables of $Q$ that makes all the atoms in the body of $Q$ true the atom $h(head(Q))$ is in $Q(D)$. 

{\em Conjunctive queries with arithmetic comparisons (CQAC for short)} are conjunctive queries that, besides the
relational subgoals, use also  subgoals that are arithmetic comparisons (AC for short), i.e., of the form
$X\theta Y$ where $\theta$ is one of the following: $<, >, \leq, \geq, =, \neq$, and  $X$ is a variable and
$Y$ is either a variable or constant. If $\theta$ is either $<$ or $>$ we say that it is an open arithmetic comparison
and if $\theta$ is either $\leq$ or $\geq$ we say that it is a closed AC. If the AC is either of the form $X<c$ or $X\leq c$ (either $X>c$ or $X\geq c$, respectively), where $X$ is a variable and $c$ is a constant,  then it is called \textit{left semi-interval}, LSI for short (\textit{right semi-interval}, RSI for short, respectively).
In the following, we use the notation $Q=Q_0+\beta$ to describe a CQAC query $Q$, where
$Q_0$ are the relational subgoals of $Q$ and $\beta$ are the arithmetic comparison
subgoals of $Q$. We define the {\em closure} of a set of ACs to be all the ACs that are implied by it.

The result $Q(D)$ of a CQAC $Q$, when it is applied on a database $D$, is given by considering all the assignments of variables (in the same fashion as  in CQs) such that the  atoms in the body  are included in $D$ and the ACs are true. For each such assignment, we produce a fact in the output  $Q(D)$.

All through this paper, we assume the following setting for a CQAC:
\squishlist
	\item [1.] Values for the arguments in the arithmetic comparisons are
	chosen from an infinite, totally densely ordered open interval, such as the
	rationals or reals.
	\item [2.] The arithmetic comparisons are  consistent, i.e., they are satisfiable.
	\item [3.] All the comparisons are safe, i.e., each variable in the
	comparisons also appears in some relational subgoal.
\squishend

A \textit{union of CQs} (resp. CQACs) is defined by a set $\CQ$ of CQs (resp. CQACs) whose heads have the same arity, and its answer $\CQ(D)$ is given by the union of the answers of the queries in $\CQ$ over the same database instance $D$; i.e., $\CQ(D)=\bigcup_{Q_i\in\CQ}Q_i(D)$.

A query $Q_1$ {\em is contained} in a query $Q_2$, denoted
$Q_1 \sqsubseteq Q_2$, if for any database $D$ of the base
relations, the answer  computed by $Q_1$ is a subset of the answer
computed by $Q_2$, i.e., $Q_1(D) \subseteq Q_2(D)$. The two
queries are {\em equivalent}, denoted $Q_1 \equiv Q_2$, if $Q_1
\sqsubseteq Q_2$ and $Q_2 \sqsubseteq Q_1$.

A {\em homomorphism} $h$ from a set of relational atoms $\CA$ to another set of relational atoms $\CB$ is a mapping
of variables and constants from one set to variables or constants of the other set that maps
each variable to a single variable or constant and each constant to the same constant.
Each atom of the former set should map to  an atom of the latter set with the same relational symbol. 

A {\em containment mapping} from a conjunctive query $Q_1$ to a conjunctive query $Q_2$  is a homomorphism from the atoms in the body of $Q_1$ to the atoms in the body of $Q_2$ that maps
the head of $Q_1$ to the head of $Q_2$.  All the mappings we refer to in this paper are containment mappings unless we say otherwise. Chandra and Merlin \cite{Chandra77-1} show that a conjunctive query $Q_2$
is contained in another conjunctive query $Q_1$ if and only if
there is a containment mapping from $Q_1$ to $Q_2$. The query containment problem for CQs is NP-complete.


\subsection{Testing query containment for CQACs}
\label{subsec:canonical-dbs}
In this section, we describe two   tests for CQAC query containment; using containment mappings and using canonical databases. 
%

First, we present the test using containment mappings (see, e.g., in \cite{2017Afrati}). Although finding a single containment mapping suffices to test query containment for CQs (see the previous section), it is not enough in the case of CQACs. In fact, all the containment mappings from the containing query to the contained one should be considered. Before we describe how containment mappings can be used in order to test query containment between two CQACs, we define the concept of normalization of a CQAC.

\begin{definition}
	\label{dfn-normalization}
	Let $Q_1$ and $Q_2$ be two conjunctive queries with arithmetic
	comparisons (CQACs).  We want to test whether $Q_2
	\sqsubseteq  Q_1$. To do the testing, we  first  normalize
	each of $Q_1$ and $Q_2$ to $Q'_1$ and $Q'_2$, respectively. We {\em normalize} a CQAC query as
	follows:
	
	\squishlist
		\item For each occurrence of a shared variable $X$ in a normal (i.e., relational) subgoal,
		except for the first occurrence, replace the occurrence of $X$ by a fresh
		variable $X_i$, and add $X = X_i$ to the comparisons of the
		query; and
		
		\item For each constant $c$ in a  normal subgoal,  replace the constant by a
		fresh variable $Z$, and add $Z = c$ to the comparisons of the
		query.
	\squishend
\end{definition}

Theorem~\ref{thm:cont-CQAC}\cite{Gupta94-1,Zhang94-4}  describes how we can test  query containment of two CQACs using containment mappings.

\begin{theorem}
	\label{thm:cont-CQAC}
	
	Let $Q_1,Q_2$ be CQACs, and $Q'_1=Q'_{10}+\beta'_1 ,Q'_2=Q'_{20}+\beta'_2$
	be the respective queries after normalization.
	Suppose there is at least one containment  mapping from $Q'_{10}$ to $Q'_{20}$.
	Let $\mu_1, \ldots , \mu_k $ be all
	the containment mappings  from $Q'_{10}$ to $Q'_{20}$. Then $Q_2
	\sqsubseteq  Q_1$ if and only if the following logical implication $\phi$
	is true:
	$$\phi: \beta'_2  \Rightarrow \mu_1(\beta'_1) \vee \cdots \vee
	\mu_k(\beta'_1).$$
	(We refer to $\phi$  as the {\em containment entailment} in the rest of this paper.)
\end{theorem}

However, in special cases the following test is also valid: 
We consider, in each query, all the equations that are derived by its set of ACs and we replace all variables that are shown  to be equal with a single variable or constant. Then we  find all mappings (which, in general, are fewer in number because there are fewer variables).  The containment entailment that uses these mappings is guaranteed, in such spacial cases, to prove query containment. 
For such special cases we say that normalization is not necessay.

The second  containment test for CQACs uses \textit{canonical databases} (see, e.g., in \cite{2017Afrati}). Considering a CQ $Q$, a canonical database is a database instance constructed as follows. We consider an assignment of the variables in $Q$ to constants such that a distinct constant which is not included in any query subgoal is assigned to each variable. Then, the facts produced through this assignment define a canonical database of $Q$. Note that although there is an infinite number of assignments and canonical databases, depending on the constants selection, all the canonical databases are isomorphic; hence, we refer to such a database instance as the canonical database of $Q$.  To test whether   $Q_2\sqsubseteq  Q_1$,   we compute the canonical database, $D$, of $Q_2$ and check if $Q_2(D)\subseteq Q_1(D)$.

Extending this test  to CQACs, a single canonical database does not suffice. We construct a set of canonical databases of a CQAC $Q_2$ with respect to a CQAC $Q_1$ as follows. Consider the set $S=S_V\cup S_C$ including the variables $S_V$ of $Q_2$, and the constants  $S_C$ of both $Q_1$ and $Q_2$. Then, we partition the elements of $S$ into blocks such that no two distinct constants are in the same block. Let $\CP$ be such a partition; for each block in the partition $\CP$, we equate all the variables in the block to the same variable and, if there is a constant in the block, we equate all the variables to the constant. For each partition $\CP$, we create a number of {\em canonical databases}, one for each total ordering on the variables and constants that are present.

Although there is an infinite number of canonical databases, depending of the constants selected, there is a bounded set of canonical databases such that every other canonical database is isomorphic to one in this set.
Such a set is referred to as \textit{the set of canonical databases} of $Q_2$ with respect to $Q_1$. To test whether $Q_2\sqsubseteq  Q_1$, we construct all the canonical databases of $Q_2$ with respect to $Q_1$ and, for each canonical database $D$, we check if $Q_2(D)\subseteq Q_1(D)$.

\begin{theorem}
	A CQAC query $Q_2$ is contained into a CQAC query $Q_1$ if and only if, for each canonical database, $D$, of $Q_2$ with respect to $Q_1$, $Q_2(D)\subseteq  Q_1(D)$.
%
\end{theorem}

\subsection{Answering queries using views}

A \textit{view} is a named query which can be treated as a regular relation. The query defining the view is called \textit{definition} of the view (see, e.g., in \cite{2017Afrati}). 

%
%
%
 Considering a set of views $\V$ and a query $Q$ over a database schema $\CS$, we want to answer $Q$ by accessing only the instances of views~\cite{Levy95-1, Halevy01-1, 2017Afrati}.  To answer the query $Q$ using $\V$ we could rewrite $Q$ into a new query $R$ such that $R$ is defined in terms of views in $\V$ (i.e., the predicates of the subgoals of $R$ are view names in $\V$). 
We denote by $\V(D)$ the output of applying all the view definitions on a database instance $D$. Thus, $\V(D)$ and any subset of it defines a view instance $\I$ for which there is a database $D$ such that $\I \subseteq \V(D)$.

If,  for every database instance $D$, we have $R(\V(D))=Q(D)$ then $R$ is an \textit{equivalent rewriting} of $Q$ using $\V$. If $R(\V(D))\subseteq Q(D)$, then $R$ is a \textit{contained rewriting} of $Q$ using $\V$. 

\begin{definition}
	A rewriting $R$ is called a {\em maximally contained rewriting} ({\em MCR}) of query $Q$ using views $\V$ with respect to query language $\CL$ if 
\squishlist
		\item [1.] $R$ is a contained rewriting of $Q$ using $\V$ in $\CL$, and 
		\item [2.] every contained rewriting of $Q$ using $\V$ in language $\CL$ is contained in $R$.
\squishend
\end{definition}

Other concepts, like Datalog queries, and computing certain answers  of a query given a view instance, $\I$, will be defined in the sections where they appear.

\section{Reasoning with Arithmetic Comparisons}
\label{sec-reason-AC}

\subsection{Notation in Presentation of Results}

We present our results using the following notation:
We use $var$ and $const$ to denote any variable or any constant.
We define:
\squishlist
\item {\em Semi-interval (SI for short)} arithmetic comparisons are of the form $var\leq const$, $var<const$
(these are called {\em left semi-interval} arithmetic comparisons, {\em closed} and {\em open} respectively ) or of the form $var\geq const$, $var>const$
(these are called {\em right semi-interval} arithmetic comparisons, {\em closed} and {\em open} respectively).

For short, we use the notation CLSI for closed left semi interval arithmetic comparisons, ORSI for open right arithmetic comparisons, and similarly, we use ORSI, CLSI, or RSI, LSI if we refer to both closed and open.

%
\squishend

\begin{definition}
\label{dfn-ac-type}
An {\em AC-type} is one of the elements of the following set $T_{AC}$:
{\small
$$T\!_{AC}\!= \!\{var\!\leq \!var, var\!<\!var,  var\!\leq \!const, var\!<\!const,   const\!\leq\! var, $$ $$const\!<\!var, var\!=\!var, var\!=\!const, var\!\neq \!var, var\!\neq \!const \}$$}
Let $\theta$ be one of the $\{ <,>,\leq, \geq, =,\neq \} $.
We say that an AC $X\theta  Y$ is of type  ''$var~\theta~var$''  if both $X$ and $Y$ are variables. If $X$ is a variable and $Y$ is a constant then we say that  it is of type ''$var~\theta~const$.''
\end{definition}

For example, a closed LSI AC is of type  $var\leq const$. 

An {\em AC-family}, $T_A$ is defined by a subset of $T_{AC}$. An AC belongs to a specific AC-family if 
it is of the  type that belongs in the family. 

Let $T_A$ be an AC-family. Then $T_A$  defines a class, $\mathcal{Q}$, of CQAC queries as follows: A query $Q$ belongs in $\mathcal{Q}$ 
 if $Q$ uses ACs only of types  in $T_A$. 
\begin{table*}
  \footnotesize{ \begin{center} \begin{tabular} {|c|c|} \hline
{\bf Notation} & {\bf Meaning}  \\\hline \hline
 CQ&  conjunctive query \\\hline
AC& arithmetic comparison  \\\hline
CQAC& conjunctive quer with   arithmetic comparison\\\hline
SI &  semi-interval AC \\\hline
LSI (RSI)& left semi-interval AC (right semi-interval AC) \\\hline
OLSI (ORSI)&   open  left (right) semi-interval AC\\\hline
CLSI (CRSI)&   closed  left (right) semi-interval AC\\\hline
OSI (CSI)& open (closed) AC  \\\hline
CQSI or CQAC-SI & conjunctive query with semi-interval ACs  \\\hline
RSI1&  set of (or conjunctive query with)  SI ACs of which only one is RSI \\\hline
$var$ ($const$)&  variable (constant) \\\hline
lhs (rhs)& left hand side (right hand side) \\\hline
Datalog-expansion& The CQ (CQAC) that results from unfolding the rules of a Datalog query  \\\hline
view-expansion& The CQ (CQAC) that results by replacing the subgoals of a rewriting with   \\\hline
&the view definitions  \\\hline

\end{tabular}

\caption{Notation and Abbreviations. }

\label{fig:table-notation}
\end{center} }  \end{table*}

Table \ref{fig:table-notation} explains the notation. It is not exhaustive. The notation that is missing follows the same pattern. 
In  Table \ref{fig:table-results} we present the results using abbreviations (e.g., LSI,CLSI) whereas in the corresponding theorems, in order to present the results in a homogenous  way, we define classes of queries as an AC-family.

\subsection{ Computing the closure of a set of ACs}

%

We list below a  set of  eight {\em elemental implications} which can be used to derive any 
AC, $b$, from any consistent set of  ACs, $F$ which does not include equalities. However, when there are equalities in $F$, we incorporate them in the rest of ACs, i.e., if $X=Y$ is in $F$, we rewrite the ACs by substituting $X$ for $Y$.

{\small
\squishlist
\item [ 1.] $\mathsf{\mathsf{\mathsf{X}}} \leq \mathsf{\mathsf{X}}$
\item [ 2.] $ \mathsf{\mathsf{X}}<\mathsf{Y} \Rightarrow \mathsf{X}\leq \mathsf{Y}$
\item [ 3.] $\mathsf{X}<\mathsf{Y} \Rightarrow  \mathsf{X}\neq \mathsf{Y}$
\item [ 4.] $ 
\mathsf{X}\leq \mathsf{Y} ~\wedge~ \mathsf{X}\neq \mathsf{Y} \Rightarrow \mathsf{X}<\mathsf{Y}$

\item [ 5.] $
\mathsf{X}\neq \mathsf{Y} \Rightarrow \mathsf{Y}\neq \mathsf{X}$
\item [ 6.] $
\mathsf{X}<\mathsf{Y} ~\wedge~ \mathsf{Y}< \mathsf{Z} \Rightarrow \mathsf{X}<\mathsf{Z}$
\item [ 7.] $
\mathsf{X}\leq \mathsf{Y} ~\wedge~ \mathsf{Y}\leq  \mathsf{Z} \Rightarrow \mathsf{X}\leq \mathsf{Z}$
\item [ 8.] $\mathsf{X}\leq \mathsf{Z} ~\wedge~ \mathsf{Z}\leq \mathsf{Y} ~\wedge~ 
\mathsf{X}\leq \mathsf{W}~\wedge~ \mathsf{W}\leq \mathsf{Y} ~\wedge~ \mathsf{W}\neq \mathsf{Z}\Rightarrow \mathsf{X}\neq \mathsf{Y}$
\squishend

where, $\mathsf{X} ,\mathsf{Y}, \mathsf{Z},\mathsf{W}$ can be either variables or constants as follows:  In (1), $\mathsf{X}$ is a variable, in (2), (3), (4) and  (5), if $\mathsf{X}$ is a variable then $\mathsf{Y}$ is either a variable or constant and vice versa. In (6) and (7), either $\mathsf{X}$ is a variable or $\mathsf{Y}$ is a variable.  In (8), 
 one of the $\mathsf{X}$ or $\mathsf{Y}$ is a variable and the rest are variables or constants. When we apply the elemental implications, we assume that,  all the obvious ACs between two constants (e.g., $5\neq 6$ or $5<6$) 
are included in $F$.
%
}

 The proof of soundness and completeness for this set of elemental implications  when there are no constants can be found in \cite{Ullman88}. 
Now we prove:
\begin{theorem}
 The elemental implications is a sound and complete set  of implications which can be used to derive any 
AC, $b$, from any consistent set of  ACs, $F$ in the presence of  constants too. 
\end{theorem}

\begin{proof}
Let $F$ be  a set of ACs that use constants too. 
For each constant $c$, we create a variable $X^c$. We view them as   regular variables, but, for ease of reference, we call these new variables, pseudo-variables. 
We construct a set $F'$ of ACs which is $F$ where each variable $c$ is replaced by $X^c$ and $F'$ contains also a  new set 
of ACs that enforce a total ordering among the pseudo-variables which is the same as their counter parts (i.e., the constants which they replaced). $F'$ also constains disequations $\mathsf{X^c}\neq \mathsf{Y^c}$ for every pair of constants $\mathsf{X}\neq \mathsf{Y}$.
The result is an immediate consequence of the following lemma:

\begin{lemma}
Let $F$ be a set of ACs that use constants. Consider the above construction of $F'$. 
There is a truth assignment for $F$ iff there is a truth assignment for $F'$. 
\end{lemma}
\begin{proof}
The one direction is trivial. For the other direction, let $M$ be a  truth assignment for $F'$. 
 We will construct a truth assignment $M'$ for $F$.

 A pseudo-variable, $X^c$, is  assigned its constant, $c$ in $M'$, i.e., $M'(X^c)= c$. Let $I$ be the set of all constants used in $F$. We consider the total ordering among the constants in $I$ and for each pair of consecutive constants $c_1, c_2$ we do the following:

Let $X_1, X_2$ be the pseudo variables corresponding to constants $c_1, c_2$. We consider the interval $(M(X_1) M(X_2))$. For all the variables that are assigned values in this interval by $M$, we assign values in $M'$, as follows. We assign any values in the interval $c_1, c_2$ that retains their ordering in $M$.

Thus the variables have the same order in $M$ and $M'$. Moreover, the constants in $F$ are positioned in the same order with respect to these variables as with respect to the pseudo-variables.\end{proof}\end{proof}
%
%


  
%

\begin{definition}
We say that a {\em variable $X$  and a constant $c$ are related (or $X$ is related to $c$ and vice versa)} in a logically closed set $F$ if there is  $X\theta c$ in $F$.
\end{definition}

When there are semi-interval ACs in the set, $F$, of ACs, then the closure is computed with respect to the set of constants that appear in $F$. I.e., only SIs that use constants from this set are included in the closure.
The following lemma summarizes some easy observations  and we will, conveniently, refer to them often in the rest of the paper.

\begin{lemma} 
\label{lemm-3clauses}
\squishlist
\item [1.]
Suppose a set, $F$, of ACs includes $W\neq Z$. Suppose 
 elemental
 implication (8) is applied on $F$  and derives $X\neq Y$  by using $W\neq Z$. Suppose, on a second step, we apply (8) on $F\cup \{ X\neq Y \}$  to derive $X_1\neq Y_1$.  Then $X_1\neq Y_1$ can be derived from $F$ by using (8) only in one step.

\item [2.]
 If $X\leq Y$ is in the closure of a set of ACs, $F$, 
then, it can be derived by using several times  elemental implication  (2) and, after that, by using several times  elemental implication (7). The application of (7) several times creates a {\em chain,} $X_1\cdots X_k$, of variables/constants where, for each 
$i=1,\ldots, k-1$, either $X_i \leq X_{i+1}$ is in $F$ or $X_i < X_{i+1}$ is in $F$.

\item [3.]
 If $X< Y$ is in the closure of a set of ACs, $F$, 
then, either there is a chain of vars/const related by $<$ from $X$ to $Y$ or, there are two chains from $X$ to $Y$  (not necessarily distinct)  of only variables  (except $X$ or $Y$, one of which could be a constant) related by $\leq$  and there are  $Z,W$, one in each chain, such that $Z\neq W$. 

\item  [4.]  Suppose  there is a  chain from $var/const$ $X$ to  $var/const$ $Y $ as in the case (2) and (3)  above.  Then there is  a  chain from $var/const$ $X$ to  $var/const$ $Y $ that contains at most two constants.
%
%
%
%

\squishend
\end{lemma}

The proof is in  \ref{app-A-A}.

\begin{lemma}
\label{lem-closure}
We can compute the closure of a set of ACs in time polynomial on the number of ACs.
\end{lemma}

All through the paper, when we refer to a set of ACs, $F$, we mean  a logically closed set of ACs.

%
%
%
%
%
%
%
%
%
%
%
%
%

In  this paper, we will often focus  on an implication of the form 
\begin{equation}
a_1~\wedge~ a_2 ~\wedge~ \cdots ~\wedge~ a_n \Rightarrow b_1\vee b_2 \vee \cdots \vee b_m \\ \label{aa}
\end{equation}
where   $a_i$'s   and $b_i$ 's  are ACs.  We define a {\em containment implication} to be a logical implication where the lhs is a conjunction of ACs and the rhs is a disjunction of ACs.
For example,      $X > 5 \Rightarrow  X < 3 \vee  X < 4$   is a  containment implication, but     $X > 5 \Rightarrow  X < 3 \wedge  X < 4$ is not because the latter uses a conjunction on the right hand side.

We say that the containment implication is {\em minimal}  or is {\em in minimal form} if the following is true: If we delete any disjunct from the rhs then the implication is not true. In  \ref{app-A-A}, a series of lemmas are presented concerning  the containment implication and the containment entailment that will be useful in the proofs of the results in this paper.

%
\section{$\Pi^p_2$-hardness Result}
\label{sec-hardd}
In this section we prove a series of theorems that refer to hardness results when only semi-interval ACs are used in the containing query. Even if the contained query uses only semi-interval ACs and an equation AC, the problem is proven here to be $\Pi^p_2$-hard.
The proof is the same as the proof in \cite{Meyden97-12} only it uses different gadgets to prove hardness results for different classes of queries, specifically queries that restrict their ACs to be only semi-interval, extending, thus, the results in  \cite{Meyden97-12}. The gadgets will be explained in detail in the proof. 

We begin with the following theorem which proves $\Pi^p_2$-hardness even in the case the contained query uses only ACs  that compare a variable to a constant.
\begin{theorem}
\label{thm-hardness}
The following problem is $\Pi^p_2$-hard: Given CQAC query $Q_2$ defined by the  AC-type $\{var\!< \!const, var\!> \!const,    var\!\neq  \!const
 \}$ and CQAC query $Q_1$ defined by the  AC-type $\{var\!< \!const, var\!> \!const
 \}$.
determine whether $Q_2\sqsubseteq Q_1$.
\end{theorem}


%
%
%
%

\begin{proof}
The reduction will be done from the $\Pi_2{-}SAT$ problem which is the following:

Instance: A $\Pi_2$ formula of quantified propositional logic, i.e., an expression 
\begin{equation}
\label{eq-PiSAT}
\forall p_1\cdots p_n\exists q_1\cdots q_m[\psi]
\end{equation}
where $\psi$ is a formula of propositional logic containing only the variables\\
 $p_1,\ldots ,p_n, q_1,\ldots ,q_m.$

Question: Is it true that formula (\ref{eq-PiSAT}) is satisfiable, i.e., is it true that for every assignment of truth values to  $p_1\cdots p_n$  there exists  an assignment of truth values to  $q_1\cdots q_m$ such that $\psi$ is true?

We will construct two Boolean CQAC queries $Q_1$ and $Q_2$ and we will prove that formula (\ref{eq-PiSAT}) is satisfiable if and only if $Q_2\sqsubseteq Q_1$.

{\bf Construction of $Q_2$:} Query $Q_2$ contains the following subgoals over constants $e$ and $f$ which encode a Boolean computation. The constants $e$ and $f$ encode True and  False respectively, the relation $a$ encodes ''and'', the relation $o$ encodes ''or'', and the relations $not$ and $true$ encode negation and the true value respectively):
$$a(e,e,e),a(e,f,f),a(f,e,f),a(f,f,f),  o(e,e,e),o(e,f,e), $$ $$o(f,e,e),o(f,f,f), not(e,f),not(f,e),true(e)$$
The above subgoals simulate the calculation of the truth value of the formula $\psi$.
In addition, we add $n$ copies of the following five subgoals:
$$a_i(U_i,e),a_i(V_i,f),a_i(W_i,e),a_i(W_i,f),U_i<7, 7< V_i,W_i\neq 7$$
for $i=1,\ldots n$, where $U_i,V_i$ are fresh and distinct variables that appear only in one relational subgoal and $W_i$ is also a fresh distinct variable which appears in two relational subgoals. 
%
%
%

In order to explain the various variants of this proof which will produce a series of different results, we refer to these four subgoals as {\em gadget 1}.

{\bf Construction of $Q_1$:} 
First, we construct the subgoals which will map to the $a_i$s above. These subgoals are $n$ copies of the following:
$$a_i(X_{1i},X_i), a_i(X_{2i},X_i), X_{1i}<7, 7< X_{2i}$$
where the variables are fresh and distinct for every copy.
We will refer to one of these copies as {\em gadget 2}.
The above formula expresses the fact that, while evaluating $Q_1$  on a canonical database of $Q_2$, each $X_i$ will take the value $e$ or $f$ and for each $e$ or $f$, there is a canonical database of $Q_2$ on which
$X_i$ takes exactly one of these two values. 
 This encodes all the combinations of truth values for the variables $p_1\cdots p_n$ in formula (\ref{eq-PiSAT}).

Finally, we construct a number of subgoals using the relations $a,o,not,true$ that will encode satisfaction of the formula 
$\psi$. 
It is based on a parsing tree for $\psi$. The construction of the parsing tree is done inductively on the structure of $\psi$. 
The leaves of the tree are labeled by  the variables $X_1,\ldots ,X_n, X_{n+1},\ldots ,X_m$ (some variables may appear in more than one leaf). Each internal node is associated with one of the symbols $\wedge, \vee, \neg$ and is labeled by a subexpression of $\psi$. If an internal node is associated with $\wedge$ and  has children labeled with subexpressions $\psi_1 $ and $\psi_2$ then the subexpression  $\psi_1 \wedge \psi_2$ labels this node. Similarly for $\vee$ and $\neg$, the latter being a node with only one child.  

Following $Tr$, we construct tree $Tr_1$ which has the same structure as $Tr$ only with different labels on the internal nodes.
The internal node of the node $\psi_1 \wedge \psi_2$ in $Tr$ is labeled in $Tr_1$ by $a(Y_1,Y_2,Y')$ where $Y_1$ is the last variable in the label of one child, $Y_2$ is the  last variable in the label of the other  child and $Y'$ is a fresh variable distinct from all the others. The subgoals that we finally add to $Q_1$ are all the atoms on the labels of the internal nodes of the tree $Tr_1$ and in the root. Moreover, we add atom $t(Y_h)$ where $Y_h$ is the last variable in the root node of the tree.

The intuition is the following:
First, the relations $a,o,not,true$ use only constants in their atoms in every canonical database of $Q_2$.
The variables in the leaves will map either to constant $e$ or to constant $f$ encoding true or false. This is enforced  by gadget 1.
The internal nodes encode truth values for the specific subexpression calculated from truth values of its children. Thus , e.g., the internal node with label $a(Y_1,Y_2,Y')$ will map to $a(e,f,f)$ encoding the fact that the subexpression (under the  specific truth assignment of the leaves) is false because the subexpression of the one child is true and the subexpression of the other child is false.

For each truth assignment to the variables $X_1,\ldots,X_n$ (that correspond to a truth assighment of $p_1,\ldots, p_n$ in  formula \ref{eq-PiSAT}), there is a canonical database  of $Q_2$ such that there is a mapping from the $a_i$s of $Q_1$ to the  $a_i$s of $Q_2$. E.g., if $p_1$'s truth assignment is True then we focus on atoms of $a_1$ and, specifically, it is a canonical database with atoms $a_1(U_1,e),a_1(V_1,f),a_1(W_1,e),a_1(W_1,f),U_1<7, 7< V_1,W_1\neq 7$
(with constants replacing the variables, of course) where we use any value such that $U_1 <7$ and $W_1 >7$.
These atoms will be used to map 
$a_1(X_{11},X_1), a_1(X_{21},X_1), X_{11}<7, 7< X_{21}$ in the mapping from $Q_1$ to this canonical database of $Q_2$.

%
%

Suppose the formula (\ref{eq-PiSAT})  is true and let $D$ be a canonical database of $Q_2$ on which  $Q_2$ and $Q_1$ evaluate to true.
By construction of $Q_2$ and especially the subgoals over the relations $a_i$, 
each $X_i $ of $Q_1$ will either map on e or f.  
The subgoals of $Q_1$ that resulted from the construction of the tree, during the evaluation, will map appropriately.
The existence of such a mapping is justified by the truth value of the formula (\ref{eq-PiSAT}) and the meaning of the subgoals built following the tree of $\psi$. Since there is a truth assignment to $q_1,\ldots ,q_m$ that make $\psi$ true, this truth assignment will enable the mappings from the labels in the internal nodes of tree $T_{tr1}$ to the canpnical database considered.

For the other direction: Suppose that on every canonical database of $Q_2$ the query $Q_1$ computes to True via mapping $\mu$. 
As we mentioned above, for every truth assignment to $p_1,\ldots, p_n$ there is a canonical database of $Q_2$ 
such that the variables $X_1,\ldots X_n$ map  via mapping $\mu$ to True or False as required by the truth assignment we assumed. 
 To prove that there are values to $q_1,\ldots q_m$ that make $\psi$ true, we need to argue inductively on the tree. Since $\mu$  is mapping the relations that label the nodes of the tree to the relations 
$a, o, not, true$  in the canonical database (recall that these relations encode logical computation) this mapping shows that the constants in the canonical database of $Q_2$ on which $X_{n+1},\ldots, X_m$ map provide the 
truth values that make the formula $\psi$ true.
%
%
Hence the formula (\ref{eq-PiSAT}) is true.
\end{proof}

Changing the gadgets as follows we obtain the same original proof of \cite{Meyden97-12} which proves hardness for the case only $<$ is allowed, but the ACs are not SIs. In the same paper it is also claimed that the hardness result also is true when only
$\leq$ is allowed. 

gadget 1: $$a_i(U_i,e),a_i(V_i,f),a_i(W_i,e),a_i(W_i,f),U_i<V_i$$

gadget 2: $$a_i(T_{1i},T_i), a_i(T_{2i},T_i), T_{1i}< T_{2i}$$

The hardness result is also valid for the case the containing query uses only open LSI comparisons.
 In particular we use, in the proof, the following two pairs of gadgets

gadget 2: $$ a_i(X_i,5,T_i),  X_i< 5$$

gadget 1: $$a_i(X_i,5,e), a_i(Y_i,X_i,f),X_i\leq 5, Y_i<5$$

gadget 2: $$ a_i(X_i,T_i),  X_i< 5$$

gadget 1: $$a_i(X_i,e), a_i(Y_i,f),X_i\leq 5, Y_i\leq5,X_i\neq Y_i$$

The first pair restricts the contained query not to use $\neq$, the second pair restricts the contained query not to use $<$. In particular, they derive the following theorem:
%
%
\begin{theorem}
\label{thm-hardness123}
1. The following problem is $\Pi^p_2$-hard: Given CQAC query $Q_2$ defined by the  AC-type $\{var\!< \!const, var\!\leq  \!const, var\!= \!const   
 \}$ and CQAC query $Q_1$ defined by the  AC-type $\{var\!< \!const, var\!= \!const 
 \}$
determine whether $Q_2\sqsubseteq Q_1$.

2. The following problem is $\Pi^p_2$-hard: Given CQAC query $Q_2$ defined by the  AC-type $\{var\!\leq  \!const, var\!\leq  \!const,    var\!\neq  \!var
 \}$ and CQAC query $Q_1$ defined by the  AC-type $\{var\!< \!const
 \}$
determine whether $Q_2\sqsubseteq Q_1$.

\end{theorem}

The following two sets of subgoals prove the case when only $\neq$ is used which is the results obtained in \cite{Kolaitis98-1}.

gadget 2: $$a_i(X_i,Y_i,T_i),  X_i\neq Y_i$$

gadget 1: $$a_i(X_i,Y_i,e),a_i(Y_i,Z_i,f),X_i \neq    Z_i$$

\section{Complexity of query containment. RSI1 queries}
\label{sec-si-no-norm}

Earlier work \cite{Klug88-1},\cite{Afrati04-1},\cite{AfratiLM06} has used the homomorphism property to prove membership of a CQAC query containment problem in NP.

\subsection{Homomorphism Property (HP)}

\begin{definition}
{\em Homomorphism property (HP)} 
Let $\mathcal{Q}_1$, $\mathcal{Q}_2$ be two classes of CQAC queries. We say that the homomorphism property holds from $\mathcal{Q}_1$ to $\mathcal{Q}_2$ if, for any pair of queries $Q_1 \in \mathcal{Q}_1$,  $Q_2 \in \mathcal{Q}_2$ the following are equivalent:

1. $Q_2 \sqsubseteq Q_1$.

2.  There is a homomorphism $\mu$ from the relational subgoals of $Q_1$ to the relational subgoals of $Q_2$ 
such that the following is true:  $\beta_2 \Rightarrow \mu(\beta_1)$.
\end{definition}

The  next theorem  is a straightforward consequence of Lemma \ref{lem-forgott} and  Lemma  \ref{lemma-direct} and extends results derived by Klug.

%
%

\begin{theorem}
\label{thm-HP}
(Complexity of query containment by homomorphism property(HP))
Consider the cases in 
Lemma \ref{lemma-direct}, where $A$ defines the set of ACs allowed in queries in the class  $\mathcal{Q}_2$ and $B$ the set of ACs allowed in  the class  $\mathcal{Q}_1.$ 

%
Then, the following problem is in $NP$:
 Given queries $Q_1 \in \mathcal{Q}_1$ and any CQAC $Q_2$, determine whether $Q_2 \sqsubseteq Q_1.$

\end{theorem}
In the rest of this paper we will go beyond homomorphism property, i.e., we will present results for cases  where the problem of deciding query containment is in NP but the homomorphism property does not hold. In this section, we consider containing queries that do not use equalities in their normalized form. In the next two sections, we 
will extend to when equalities are also present, and we will show cases where normalizations is not necessary to prove query containment via the containment entailment test.

\subsection{RSI1 queries}
We consider RSI1 queries:
\begin{definition}
An RSI1 query is a CQAC query which uses LSI ACs (open or closed) and exactly one RSI AC (open or closed).
\end{definition}
In our notation, we  use  the superscript ``$^{one}$'' to denote that only one RSI is allowed, thus, the class of closed RSI1 queries is  defined by the AC-type
 $\{var\!\leq \!const, var\!\geq \!const^{one} \}$.

%
%
%
The following proposition will be used often:

\begin{proposition}
	\label{trick-pro}
	Let $\beta$ be a conjunction of  ACs, and let each $\beta_1, \beta_2, \ldots ,\beta_k$  be a conjunction of a set of closed
	SIs of which only one is RSI. Suppose the following is true:
	$\beta \Rightarrow \beta_1 \vee \beta_2 \vee \cdots \vee \beta_k.$
	Then there is a $\beta_i$ (w.l.o.g. suppose it is $\beta_1$) such that either  of the following two happens:
	
	%
	%
	%
	%
	%
	%
	\begin{enumerate}[label=(\roman*)]
		\item $\beta \Rightarrow \beta_1,$ \textbf{or}
		\item there is an AC, $e$, among the conjuncts of  $\beta_1$, such that $\beta \not \Rightarrow e$  and for each AC, $e'$, in $\beta_1$ such that $e'\neq e$, we have that $\beta  \Rightarrow e'.$ 
	\end{enumerate}
\end{proposition}
\begin{proof}
	Suppose there is no $\beta_i$ such that $\beta \Rightarrow \beta_i$.
	Towards  contradiction of clause (ii) in the statement of the proposition,
	suppose that for all the $\beta_i$s there are at least two ACs that are not implied by $\beta$. Since each
	 $\beta_i$ uses at least one LSI   and only one RSI, there is at least one LSI in $\beta_i$  that is not implied by $\beta$.  Let $\gamma$ be the disjunction of all such ACs.  If the implication in the premises of this proposition is true, then, $\beta \Rightarrow \gamma$ is also true (just apply the distributive law). 
This is impossible due to
Lemma \ref{lemma-direct}, clause (1).
\end{proof}

\begin{proposition}
 \label{trick-proO} 
Proposition 
\ref{trick-pro} 
is  true for  each of the following cases:
\squishlist
\item[1.]  $\beta$ is a set of ACs and each disjunct contains closed LSIs and one open RSI.
\item[2.]  $\beta$ is a set of ACs and each disjunct contains open LSIs and one open RSI.
\item[3.]  $\beta$ is a set of ACs and each disjunct contains open LSIs and one closed RSI.
\squishend
Moreover, for the cases (2) and (3), the following conditions should be satisfied:
For 
any $X\neq Y$ that appears in  $\beta$,     if a constant (say $c_0$) is related by an AC to both $X$ and $Y$  in $\beta$   then, either (i) $c_0$ does not relate to both $X$ and $Y$ by a closed AC in   $\beta$ or  (ii) $c_0$ does not appear in an  open AC in   some $\beta_i$. 

%
\end{proposition}
The proof of the above proposition is the same as the proof of Proposition 
\ref{trick-pro}  
except that we need to use Lemma 
\ref{lemma-direct}, clause (3) too. The above propositions  motivate the following definition which we use in Theorem \ref{lem-con-aa-aa}.

\begin{definition}
\label{dfn-con-aa-aa}
Let $F$ be a set of ACs  and let $\beta_1$ be the conjunction of ACs in $F.$ Let $S_{var}$ be a set of variables disjoint from the set  of variables in the ACs in $F.$ 

%

We say that  $F$ is an {\em NP-enabling set of ACs} if the following is true for any $\beta_2$ which is a conjunction of ACs  over $S_{var}$ and for every 
set of mappings $\mu_1, \ldots, \mu_k$  from the variables in the set of ACs in $F$ to the variables in $S_{var}$:

If the   entailment  $\beta_2  \Rightarrow  \mu_1 (\beta_1) \vee \cdots \vee \mu_k(\beta_1)$
 is true  then, 
there is a  mapping (say $\mu_1$)  such that the following happens:  There is an AC, $e$, among the conjuncts of  $\mu_1(\beta_1)$, such that $\beta_2 \not \Rightarrow e$  and for each AC, $e'$, in $\mu_1(\beta_1)$ such that $e'\neq e$, we have that $\beta_2  \Rightarrow e'.$ 
\end{definition}

The following theorem will be used to prove each of the four cases of Theorem 
\ref{thm-main-two}. 

\begin{theorem}
\label{lem-con-aa-aa}

%


Let  $\mathcal{Q}_1$ be a class of queries such that, for each $Q_1 \in \mathcal{Q}_1$, its set of ACs is a  NP-enabling set of ACs.
Then, the following problem is in $NP$:
 Given queries $Q_1 \in \mathcal{Q}_1$ and any CQAC $Q_2$, determine whether $Q_2 \sqsubseteq Q_1$.
\end{theorem}

\begin{proof}
Suppose $Q_1\sqsubseteq Q_1$ and this  is proven by the following containment entailment:
$\beta_2  \Rightarrow  \mu_1 (\beta_1) \vee  \mu_2 \beta_1) \vee \cdots \vee \mu_k(\beta_1).$ From the premises of this theorem, we deduce that
we can write the containment entailment, equivalently as:
$\beta_2  \wedge \neg e\Rightarrow  \mu_2 (\beta_1) \vee \cdots \vee \mu_k(\beta_1).$  
For this  latter entailment, according to the premises of the theorem, there is a containment mapping (say $\mu_2$) such that there is $e_1$ 
in $\mu_2(\beta_1)$ which satisfies the conditions in the statement of the theorem. Hence the containment 
entailment is, further, equivalently, written as:
$$\beta_2  \wedge \neg e \wedge \neg e_1 \Rightarrow  \mu_3 (\beta_1) \vee \cdots \vee \mu_k(\beta_1).$$
Thus, after $i$ iterations, we will have:
$$\beta_2 \wedge \neg e_1 \wedge \neg e_2 \cdots   \wedge \neg e_i \Rightarrow  \mu_{i+1}( \beta_1) \vee \cdots \vee \mu_k(\beta_1)$$
In the lhs of the above implication, we can have only polynomially many (on the number of AC appearing in both queries)
ACs, hence we can have at most polynomially many iterations. Hence the containment entailment can contain at most polynomially many containment mappings. 


%
%
%
%
%
%
%

To prove membership in $NP$ observe that   the certificate is a containment entailment and the order on which to consider the disjuncts in this containment entailment in order to prove that this containment entailment is true. 
\end{proof}

The following theorem says that if the containing query uses LSIs and at most one RSI then the query containment  problem is in NP under certain conditions:

\begin{theorem}
\label{thm-main-two}

Suppose $\mathcal{Q}_2$ is the class of CQAC queries defined by the  AC-type $T_{AC} $.
Suppose, also that   either of the following happens: 
(We write  $var \geq const^{one}$ ( or $var > const^{one}$) to mean that at most one RSI is present.)

\squishlist
\item [1.]
$\mathcal{Q}_1$ is the class of CQAC queries defined by the  AC-type $\{var\!\leq \!const, var\!\geq \!const^{one}   
 \}$.

\item [2.]
$\mathcal{Q}_1$ is the class of CQAC queries defined by the  AC-type $\{var\!\leq \!const, var\! > \!const^{one}
\}$.

\item [3.]
$\mathcal{Q}_1$ is the class of CQAC queries defined by the  AC-type $\{var\! <\!const, var\!\geq \!const^{one}
 \}$.

\item [4.]
$\mathcal{Q}_1$ is the class of CQAC queries defined by the  AC-type $\{var\! < \!const, var\! > \!const^{one}
 \}$.
%
%
%
%
%
%
%
%
\squishend

Suppose the following condition  is satisfied:
For 
any $X\neq Y$ that appears in  $Q_2$,     if a constant (say $c_0$) is related by an AC to both $X$ and $Y$   then, then either (i) $c_0$ does not relate to both $X$ and $Y$ by a closed AC in   $Q_2$ or  (ii) $c_0$ does not appear in an  open AC in   $Q_1$. 

Then, the following problem is in $NP$:
 Given queries $Q_1 \in \mathcal{Q}_1$ and $Q_2 \in \mathcal{Q}_2$, determine whether $Q_2 \sqsubseteq Q_1$.



%
\end{theorem}

\begin{proof}
%

The proof  is a straightforward consequence of either  Propositions \ref{trick-pro}  or \ref{trick-proO}  and of Theorem \ref{lem-con-aa-aa}.
\end{proof}

\section{The containing query has one AC. Query containment is in NP. Normalization is not necessary}
\label{sec-one-ac-cont}

Now we examine cases where we also allow equation ACs in the normalized version of the queries. 
Thus, here, we will also consider containment implications that have equation ACs on the rhs. 
We use the elemental implications to prove equality as follows:
An equation $X=Y$ is equivalent to $X\leq Y \wedge X\geq Y$ over a totally ordered dense open interval. 
To prove that this is true, we need to argue on two points: a) We observe  that when one of the $X$ or $Y$ is a constant, $c$,  then 
the following is true: if 
$X\leq c$ ($X\geq c$ respectively) then  there is more than one $X$ in the interval such that $X\leq c$ ($X\geq c$ respectively).  This observation excludes the case that $X\leq c \Rightarrow X=c$ for certain values of $c$.
b) The second point is the following:
Since we consider a consistent set of ACs, and assuming both
$X\leq Y $ and $X\geq Y $ are true,
 proving either of these ACs  will not use elemental implication (3) and hence by proving $X\leq Y \wedge X\geq Y$ using the elemental implications we prove that $X=Y$. 

The following is a useful observation to prove our results: When we prove $X\leq Y $ from a set of ACs $F$, using the elemental implications, any disequations in $F$ are not useful in the proof. Hence, 
$\beta_2 ~\Rightarrow ~X=Y \vee Z=W$,  is equivalent to $\beta_2 \wedge \neg X=\!Y ~\Rightarrow  ~Z=W$,  which is, therefore equivalent to $\beta_2 \Rightarrow Z=W$. Thus, we have proven that if the rhs of a containment implication uses only equalities then this containment implication in minimum form has only one equality AC on the rhs.

\begin{theorem}
\label{thm-one} 

Suppose $\mathcal{Q}_2$ is the class of CQAC queries defined by the  AC-type $T_{AC} $.
Suppose, also that   
$\mathcal{Q}_1$ is the class of CQAC queries defined by the  AC-type $T_{AC} $ with the restriction that 
there is at most one AC 
that is not an equation AC.

Then, the following problem is in $NP$:
 Given queries $Q_1 \in \mathcal{Q}_1$ and $Q_2 \in \mathcal{Q}_2$, determine whether $Q_2 \sqsubseteq Q_1$.
\end{theorem}

\begin{proof}
Consider the containment entailment $$\phi: \beta_2  \Rightarrow \mu_1(\beta_1) \vee \cdots \vee
	\mu_k(\beta_1)$$
where $\beta_2 $ is the conjunction of ACs in the contained query and $\beta_1$ is the conjunction of ACs in the containing query.
$\beta_1$ is either a single AC which is not an equation AC or a conjunction of ACs of which at most one is a non-equation AC.

First, suppose $\beta_1$ is  a single AC,
 hence the right hand side of the containment entailment is a disjunction of ACs. If $m$ is the number of variables and constants in the contained query, then there are at most $O(m^2)$ different ACs that use them. Each such AC is 
created by a containment mapping $\mu_i$, and therefore $O(m^2)$ containment mappings suffice to prove the 
containment entailment true. 
Given those mappings, then, according to Lemma \ref{lem-cccppp}, 
we can check in polynomial time if the containment entailment is true because now it has the form of a containment implication.

For the general case, we assume $\beta_1=b_1\wedge d_1\wedge d_2 \wedge \cdots$ where $b_1 $ is non-equation AC and the $d_i$s are equation ACs and we argue as follows: Suppose each $\mu_j(\beta_1)$ has at least one equation AC which is not directly derived by $\beta_2$, i.e., for every $\mu_j$ there is $d_l$ such that $\beta _2 \Rightarrow \mu_j(d_l)$ is not true.  In this case, we consider the containment implication which is using only  equation ACs and, specifically those equation ACs for which $\beta _2 \Rightarrow \mu_j(d_l)$ is not true. For this containment implication to be true, at least one of the equation ACs should be implied by $\beta_2$, since if we  assume otherwise, we arrive at contradiction (by use of  Lemma \ref{lem-orac}). Therefore, there is 
a $\mu_j(\beta_1)$ (suppose this is $\mu_1(\beta_1)$)  for which all equation ACs (i.e., of the form X=Y or X=c) are directly derived by $\beta_2$.
Thus, now we have to check a new (shorter) containment entailment:
$$\beta_2  \wedge \neg\mu_1(b_1) \Rightarrow \mu_2(\beta_1)\vee \cdots \vee \mu_k(\beta_1)$$
where $b_1$ is the single non-equation AC in $\beta_1$.
Now inductively, we repeat  the above argument for polynomially many steps, until we arrive at a containment entailment for which there is a single $\mu_j(\beta_1)$ on the rhs and all ACs of this $\mu_j(\beta_1)$ are implied by the lhs. Membership in NP is proven by observing that we can guess the order in which the mappings are used in the argument above and prove in each step that the given subsequent mapping has the property that each equation is implied by the  lhs of the current entailment.
I.e., in the general step of this induction, we will consider the a containment entailment  as follows:
$$\beta_2  \wedge \neg\mu'_1(b_1)  \wedge \neg\mu'_2(b_1)\ldots  \wedge \neg\mu'_j(b_1)\Rightarrow \mu'_{j+1}(\beta_1)\vee \cdots \vee \mu'_k(\beta_1)$$
and $\mu_{j+1}$ is the subsequent mapping, for which we prove that 
$$\beta_2  \wedge \neg\mu'_1(b_1)  \wedge \neg\mu'_2(b_1)\ldots  \wedge \neg\mu'_j(b_1)\Rightarrow \mu'_{j+1}(d_1)~~~~~for~all~i$$ 
\end{proof}
%


\section{Normalization}
\label{sec-normaliz}


%
In this section, first we prove  in Lemmas  \ref{lem-111}  and \ref{lem-1212}  that, for certain cases normalization is not needed to prove CQAC containment.
Then, we extend the results in previous sections by leveraging Lemmas  \ref{lem-111}  and \ref{lem-1212}. We observe that, when normalization is not necessary, then the containment entailment (on the non-normalized queries, now)  does not 
use equations, which is what we assumed in previous sections.

The following is an easy result:

\begin{theorem}
Suppose $\mathcal{Q}$ is the class of CQAC queries defined by the  AC-type $T_{AC} $ in their normalized form.
Suppose $Q_1,Q_2$ are two queries from  $\mathcal{Q}$. Suppose there is an AC in $Q_1$ which is $X=c_0$ and there is no AC in $Q_2$ which is $Y=c_0$. Then the following is true: $Q_2 \not \sqsubseteq Q_1$.
\end{theorem}
\begin{proof}
After applying the distributive law on the containment entailment, there is a containment implication which on the lhs is a disjunction of equations, all of which use the same constant, $c_0$. Since this containment implication is true iff there is a disjunct implied by the lhs, this is a contradiction, since we assumed that the lhs does not include $Y=c_0$ in the closure of its ACs. 
\end{proof}

In the rest of this  section, we find cases when normalization is not necessary.

The following is a technical lemma which is useful in the proofs of the results that follow.

\begin{lemma}
\label{lem-eq-directly}
Consider the containment entailment $$\beta_2\Rightarrow \mu_1(\beta_1)\vee  \mu_2(\beta_1) \cdots  \mu_i(\beta_1)\cdots $$
where $i\geq 2$ and  $\beta_2 $ is the conjunction of ACs in the contained query and $\beta_1$ is the conjunction of ACs in the containing query.

Let $\beta_1=b_1\wedge \cdots \wedge b_m.$ If the above containment entailment is true, then
there is a $\mu_i$ (say it is $\mu_1$) for which, each  $\mu_1(b_j)$ which  is an equation (say the equation $X=Y$) is such that $\beta_2\Rightarrow \mu_1(b_j)$, i.e.,  $\beta_2 \Rightarrow \mu_1(X) =\mu_1(Y) $.
\end{lemma}
\begin{proof}
Suppose that the conclusion in the statement of the lemma is not true. Then for all $i$, there is a $k_i$ such that    $\mu_i(b_{k_i})$ is not implied by
$\beta_2$, where $b_{k_i}$ is an equation. Then we consider the containment implication $\beta_2 \Rightarrow {\large \vee}_i \mu_i(b_{k_i})$ in a minimal form. We move (according to Lemma \ref{lem-orac}) disjuncts from the rhs to the lhs  of this containment implication to prove a certain equation (say the equation $X=Y$).  By moving, according to Lemma \ref{lem-orac}, an equation from the rhs to the lhs, we apply negation, hence we are adding a conjunct on the lhs which uses $\neq$.  Now, observe  that the elemental implications that prove an equation do not use 
$\neq$ hence, the equation is implied by $\beta_2$, which is a contradiction to the assumption made.
\end{proof}
%
\subsection{Closed ACs in the containing query}
The following lemma is an extension of a known result which had assumed that both queries use closed ACs:


\begin{lemma}
\label{lem-111}
Let $Q_1,Q_2$ be CQACs such that  $Q_1$ contains only closed ACs,  and $Q_2$ contains any AC. Let $Q'_1=Q'_{10}+\beta'_1 ,Q'_2=Q'_{20}+\beta'_2$
	be the respective queries after we have  incorporated in the relational  subgoals all equalities that appear  in the closure of the query, i.e., if $X=Y$ is in the closure of the ACs in the  query we rewrite the query by substituting $X$ for $Y$.

	Suppose there is at least one containment  mapping from $Q'_{10}$ to $Q'_{20}$.
	Let $\mu_1, \ldots , \mu_k $ be all
	the containment mappings  from $Q'_{10}$ to $Q'_{20}$. Then $Q_2
	\sqsubseteq  Q_1$ if and only if the following logical implication $\phi$
	is true:
	$$\phi: \beta'_2  \Rightarrow \mu_1(\beta'_1) \vee \cdots \vee
	\mu_k(\beta'_1).$$
\end{lemma}
\begin{proof}
Consider the containment entailment $$\beta_2\Rightarrow \mu_1(\beta_1)\vee  \mu_2(\beta_1) \cdots  \mu_i(\beta_1)\cdots $$
where $\beta_2 $ is the conjunction of ACs in the contained query and $\beta_1$ is the conjunction of ACs in the containing query both in normal form.

Suppose  that containment mapping $\mu_1$ has the property of Lemma \ref{lem-eq-directly}.
%

 Let $\beta_1=b_1\wedge \cdots \wedge b_m.$ 
Inductively, we assume that 
there is  a set  $S$ of $\mu_i$s such that, for each $\mu_i\in S$,  if $\mu_i(b_{j})$ is an equation then it is implied by $\beta_2$, i.e., 
$\beta_2\Rightarrow \mu_i(b_{j})$. 
Consider the entailment

\begin{equation}
\label{A1-aa}
\beta_2 \Rightarrow 
\bigvee _{all~\mu_i \in S} \mu_i(\beta_1)
\end{equation}

Apply the distributive law to it and obtain a set of containment implications.  If one of those containment implications does not use any AC which is ``='', we say that this is a NEQ containment implication.            
We observe that if each NEQ containment implication   is true then  the entailment \ref{A1-aa} is true, and, hence,  the $\mu_i$s  in $S$ prove containment of the one query to the other. Moreover we have used only mappings 
that satisfy Lemma \ref{lem-eq-directly}.
Hence we have proven the present lemma.

To finish the proof, 
 we will prove that 
either a) each NEQ containment implication   is true   or b)
there is a $\mu_i=\mu_k$ not in $S$ such that  for each  $b_j$ which is an equation AC in $\beta_1$ the following is true:  $\beta_2 \Rightarrow \mu _k(b_j).$

If (a) is not true then  there is a NEQ containment implication resulting from entailment \ref{A1-aa}, $\beta_2 \Rightarrow \delta_k$ which is not true.
This means that $\neg\delta_k \wedge \beta_2$ is consistent. Suppose the following containment implication is true
\begin{equation}
\label{B1-aa}
\beta_2 \Rightarrow
\bigvee _{all~\mu_i \in S} \mu_i(\beta_1)   \bigvee _{\mu_i \in S_1\\where~	S_1\cap S =\emptyset } \mu_i(\beta_1)
\end{equation}
Suppose (b) above is not true. Then  for all $\mu_i \in S_1$, there is $b_{k_i}$ in $\beta_1$ such that 
$\beta_2\Rightarrow \mu_i(b_{k_i})$ is not true.

Consider the containment implication resulting from entailment \ref{B1-aa} which uses, for each $\mu_i \in S_1$, such that the following is not true  $\beta_2 \Rightarrow \mu_i(b_{k_i})$, and each $\mu_i \in S$,  uses the ACs that form $\delta_k$.
Hence the following is true:
$$\beta_2 \Rightarrow \delta_k  \bigvee _{\mu_i \in S_1}\mu_i(b_{k_i})$$
Hence

$$\beta_2  \wedge\neg \delta_k \Rightarrow  \bigvee _{\mu_i \in S_1}\mu_i(b_{k_i})$$
 This is a contradiction because $\delta _k$ has only closed ACs, hence $\neg\delta _k$ has only open ACs, hence no AC from $\neg\delta _k$ participates  in proving equality. Therefore the above is equivalent to 
$$\beta_2  \Rightarrow  \bigvee _{\mu_i \in S_1}\mu_i(b_{k_i})$$ which is equivalent to 
$\beta_2\Rightarrow \mu_i(b_{k_i})$ for each $i$, because the $\mu_i(b_{k_i})$ are all equalities.
%
%
%
%
\end{proof}

The following example shows that if $Q_1$ contains only one open LSI and 
even if only LSIs are used in both queries, if the containing query contains a single open LSI, then, in certain cases,  normalization is necessary, so the above result does not carry over to the case open SIs appear in $Q_1$.
\begin{example}

$Q_1():- a(5,X), X<5$, \\$Q_2():- a(5,X), a(X,Y),X\leq 5, Y<5$.
\end{example}


\subsection{Only LSIs in the containing query. NP via HP}

The following is a lemma similar to lemma \ref{lem-111}:

%

\begin{lemma}
\label{lem-1212}

Suppose
$\mathcal{Q}_1$ is the class of CQAC queries defined by the  AC-type $\{var\! <\!const, var\!\leq \!const, var\! =\!const, var\! =\!var
 \}$ and $\mathcal{Q}_2$ is the class of CQAC queries defined by the  AC-type $T_{AC}$ such that the following condition is satisfied:

%
%

There is no constant that is shared by an equation AC in $Q_1$ by an open LSI in $Q_1$ and by a closed LSI in $Q_2$.

 Consider  the following problem:
 Given queries $Q_1 \in \mathcal{Q}_1$ and $Q_2 \in \mathcal{Q}_2$, determine whether $Q_2 \sqsubseteq Q_1$.
This containment problem can be decided by  checking satisfaction of a containment entailment which uses containment mappings among non-normalized queries.\footnote{recall that a CQAC query is referred to as non-normalized if its ACs do not imply equalities}

\end{lemma}
\begin{proof}
It is a copy-and-paste of the proof of Lemma \ref{lem-111} up until the last two paragraphs. Now we need to argue differently on the $\delta_k$ that appears in this proof. For that, we observe,  that an OLSI when moved on the lhs of the containment implication (as in Lemma \ref{lem-orac}) is a CRSI which, however, because of the condition cannot be used together with a CLSI to derive an equation.  
%
%
\end{proof}

The following theorem is a   consequence of Lemma \ref{lem-1212} and Theorem \ref{thm-HP}:

\begin{theorem}
\label{them-norm-lsi}
%
%
%
Suppose $\mathcal{Q}_2$ is the class of CQAC queries whose normalized version uses 
ACs from $T\!_{AC} $ and 
$\mathcal{Q}_1$ is the class of CQAC queries whose normalized version uses
ACs from $\{var\!\leq \!const, var\!\ < \!const, var\!=\!var, var\!=\!const
 \}$.

Suppose the following conditions are satisfied: 

1. For 
any $X\neq Y$ that appears in  $Q_2$,     if a constant (say $c_0$) is related by an AC to both $X$ and $Y$   then, then either (i) $c_0$ does not relate to both $X$ and $Y$ by a closed AC in   $Q_2$ or  (ii) $c_0$ does not appear in an  open AC in   $Q_1$. 

2. There is no constant that is shared by an equation AC in $Q_1$ by an open LSI in $Q_1$ and by a closed LSI in $Q_2$.
%
%

Then, the following problem is in $NP$:
 Given queries $Q_1 \in \mathcal{Q}_1$ and $Q_2 \in \mathcal{Q}_2$, determine whether $Q_2 \sqsubseteq Q_1$.

\end{theorem}
\begin{proof}
The second conditions makes sure that normalization is not necessary, and the first condition is coming from 
Theorem \ref{thm-HP}.
\end{proof}

\subsection{Both  LSIs and  RSIs appear in the containing query. Cases in NP}

The following theorem is a  straightforward  consequence of Lemma \ref{lem-111} and Theorem \ref{thm-main-two}(1).

\begin{theorem}
\label{them-norm-lsi-rsi1-closed}
Suppose $\mathcal{Q}_2$ is the class of CQAC queries whose normalized version uses 
ACs from $T\!_{AC} $ and 
$\mathcal{Q}_1$ is the class of CQAC queries whose normalized version uses
ACs from $\{var\!\leq \!const, var\!\geq \!const^{one}, var\!=\!var, var\!=\!cons
 \}$.
(We write  $var \leq const^{one}$ ( $var < const^{one}$) to mean that at most one RSI is present.)

Then, the following problem is in $NP$:
 Given queries $Q_1 \in \mathcal{Q}_1$ and $Q_2 \in \mathcal{Q}_2$, determine whether $Q_2 \sqsubseteq Q_1$.

\end{theorem}
%
%
Now, we extend Theorem \ref{them-norm-lsi-rsi1-closed} to include also open LSIs and RSIs.

\begin{theorem}
\label{them-norm-lsi-rsi1}
Suppose $\mathcal{Q}_2$ is the class of CQAC queries whose normalized version uses 
ACs from $T\!_{AC} $ and 
$\mathcal{Q}_1$ is the class of CQAC queries whose normalized version uses
either of the following:
\squishlist
\item [1.]
ACs from $\{var\!\leq \!const, var\!\geq \!const^{one}, var\!=\!var, var\!=\!cons
 \}$.

\item [2.]
ACs from $\{var\!\leq \!const, var\!  > \!const^{one}, var\!=\!var, var\!=\!cons \}$.

\item [3.]
ACs from $\{var\! <  \!const, var\!\geq \!const^{one}, var\!=\!var, var\!=\!cons \}$.

\item [4.]
ACs from $\{var\!  <  \!const, var\!  >  \!const^{one}, var\!=\!var, var\!=\!cons \}$.
\squishend

(We write  $var \leq const^{one}$ ( $var < const^{one}$) to mean that at most one RSI is present.)

We assume the following constraints: a) each constant in an RSI in $Q_1$ is different from each constant in an LSI 
in $Q_1$, b) each constant in $Q_1$ in an open SI  in $Q_1$ does not appear in a closed SI in $Q_2$.

Then, the following problem is in $NP$:
 Given queries $Q_1 \in \mathcal{Q}_1$ and $Q_2 \in \mathcal{Q}_2$, determine whether $Q_2 \sqsubseteq Q_1$.

\end{theorem}

\begin{proof}
Because of the constraints, we can use the same argument as in the proof of Lemma \ref{lem-111} (together with similar argument as in Lemma \ref{lem-1212}) to prove that 
normalization is not necessary. The rest is a consequence of Theorem \ref{thm-main-two}.
%
\end{proof}

\section{Extending the Results to Include Single Mapping Variables}
\label{sec-exte-single-mappinga-var}
In the results about complexity of query containment that we presented, we have often imposed  a condition where
a constant is not allowed to appear in certain positions. When such cases arise, it is interesting to notice that, when  the condition is not satisfied, there are still classes of queries for which the problem of query containment is in NP. 
Such classes can be derived by observing that,
often, two attributes  in a query that may be used in an arithmetic comparison represent quantities that are measured in different units, e.g., weight and height. For these attributes, say $A$ and $B$, if we have an AC $A\geq 5$ and $B> 5$, we can safely assume that these are different constants, since, in a containment mapping, an attribute that represents weight will never map on at attribute that represents 
height. These observations are further elaborated in \cite{Afrati04-1,Afrati19}.

In an orthogonal fashion, there may exist certain variables of the containing query which, in every containment mapping map on the same variable on the contained query. Such variables are the head variables, but there may be others. We call such variables {\em single-mapping} variables and the formal definition is in Definition 	\ref{defn:single-map-vars}.

We start the analysis by focusing on the head variables.
Consider two CQACs, $Q_1=Q_{10}+\beta_1$ and $Q_2=Q_{20}+\beta_2$. 
and
the containment entailment:
\begin{equation} \label{eq:cont-entail}
\beta_2\Rightarrow\mu_1(\beta_1)\vee\cdots\vee\mu_k(\beta_1)
\end{equation}
\noindent
where $\mu_1,\dots, \mu_k$ are all the containment mappings from $Q_{10}$ to $Q_{20}$.
Suppose  $\beta_1$ is such that $\beta_1 = \beta_{11}\wedge\beta_{12}$ where $\beta_{11}$ is the conjunction of ACs that use only distinguished variables and $\beta_{12}$ the conjunction of the rest of the ACs in $\beta_1$.  We observe that, in the containment entailment, each disjunct on the right hand side becomes:
\begin{align*}
\mu_i(\beta_1)=\mu_i(\beta_{11})\wedge\mu_i(\beta_{12}).
\end{align*}
However, $\mu_i(\beta_{11})$ is the same for every $i$. Thus, applying the distributive law, we write the containment entailment:
\begin{align*}
\beta_2\Rightarrow\mu_1(\beta_{11})\wedge[\mu_1(\beta_{12})\vee\cdots\vee\mu_k(\beta_{12})].
\end{align*}
Consequently, the containment entailment is equivalent to the conjunction of  the following two entailments:
\begin{align*}
\beta_{2}&\Rightarrow\mu_1(\beta_{11}). \\
\beta_{2}&\Rightarrow\mu_1(\beta_{12})\vee\cdots\vee\mu_k(\beta_{12}).
\end{align*}
\begin{definition}(single-mapping variables)
	\label{defn:single-map-vars}
	Let $Q_1=Q_{10}+\beta_1$, $Q_2=Q_{20}+\beta_2$ be two CQACs, such that there is at least one containment mapping from $Q_{10}$ to $Q_{20}$. Consider the set ${\cal M}$ of all the containment mappings from $Q_{10}$ to $Q_{20}$.  Each variable $X$ of $Q_1$ which is always mapped on the same variable of $Q_2$ (i.e., for each $\mu\in\CM$,  $\mu(X)$ always equals the same variable) is called a {\em single-mapping variable with respect to $Q_2$}.
\end{definition}

As we mentioned, the head variables of $Q_1$ are single-mapping variables with respect to any query. For another example, consider a predicate $r$ such that  $Q_1$ has the subgoals $g_{11}, g_{12},\dots , g_{1n}$ with predicate $r$ and $Q_2$ has a \textit{single} subgoal $g_2$ with predicate $r$. Since each of the $g_{11}, g_{12},\dots , g_{1n}$ subgoals maps on $g_2$, for every containment mapping from $Q_1$ to $Q_2$, the variables in $g_{11}, g_{12},\dots , g_{1n}$ subgoals are single-mapping variables.


The following theorem gives a polynomial reduction of a CQAC containment problem to one with fewer ACs in the containing query.
\begin{theorem}
\label{thm-reducttion}
Consider CQAC queries $Q_1=Q_{10}+\beta_1$ and $Q_2=Q_{20}+\beta_2$.
 Let ${\cal X}_1$ be a set of single-mapping variables of $Q_1$ with respect to $Q_2$.
Let  $\beta_1 = \beta_{11}\wedge\beta_{12}$ where $\beta_{11}$ is the conjunction of ACs that use only variables  in  ${\cal X}_1$  and $\beta_{12}$ the conjunction of the rest of the ACs in $\beta_1$. 
Let 
$Q'_1=Q_{10}+\beta_{12}$ 
Then the following are equivalent:

(i)    $Q_2 \sqsubseteq Q_1$

(ii) $\beta_{2} \Rightarrow\mu_1(\beta_{11})$ and  $Q_2 \sqsubseteq Q'_1$, where $\mu_1$ is any mapping.

\end{theorem}
%
%
%
We call the second implication in the statement of the above theorem the {\em reduced containment entailment} and we call $Q'_1$ the {\em reduced containing query with respect to $Q_2$}.

A consequence of Theorem \ref{thm-reducttion} is the following theorem:

\begin{theorem}
 Let $\mathcal{Q}_1$ and $\mathcal{Q}_2$ be two  classes of CQAC queries such that the following problem is in $NP$:
 Given queries $Q_1 \in \mathcal{Q}_1$ and $Q_2 \in \mathcal{Q}_2$, 
determine whether $Q_2 \sqsubseteq Q_1$.

Then, the following problem is in $NP$:
Given queries $Q_1$ and $Q_2 \in \mathcal{Q}_2$ where $Q_1$  is the reduced containing query of $Q'_1 \in \mathcal{Q}_1$  with respect to $Q_2$  determine whether $Q_2 \sqsubseteq Q_1$.

%
%


%
\end{theorem}

\section{Maximally Contained Rewritings and Certain Answers. }
\label{sec-mcr-intro}

In the second part of the paper we prove that, for special cases of CQAC queries, we can compute certain answers in polynomial time. We first prove that for CQAC queries and views, a maximally contained rewriting (see Definition \ref{def-mcr146}) computes exactly all certain answers. This is done in 
 Section \ref{sec-6}.     Thus, if we can find an MCR in a query language for which query evaluation has polynomial-time data complexity, then we can compute certain answers in polynomial time.
%
%
 We show in Sections \ref{sec-mcrs-cert-answe}
through  \ref{subsec-mcr-datalogAC} that, for a special case of CQAC queries with semi-interval comparisons, we can build an MCR in the language of Datalog with ACs.

\begin{definition}{\bf(Contained Rewriting)}
Given a query $Q$ and a set of views $\V$, $R$ is a {\em contained rewriting} of $Q$ in terms of $\V$ under the open world assumption if and only if, for each viewset $\I$, the following is true: For each database instance $D$ such that
$\I \subseteq \V(D)$, $R(\I)$ is a subset of or equal to $Q(D)$.
\end{definition}

\def\L{\mathcal{L}}
\def\I{\mathcal{I}}
\def\P{\mathcal{P}}
\def\D{\mathcal{D}}
\def\V{\mathcal{V}}
\def\C{\mathcal{C}}
\def\T{\mathcal{T}}
\def\S{\mathcal{S}}
\def\U{\mathcal{U}}
\def\M{\mathcal{M}}
\def\E{\mathcal{E}}
\def\K{\mathcal{K}}
\def\LAV{\textsf{LAV}}
\def\DeptA{\textsf{Dept}_\textsf{A}}
\def\DeptB{\textsf{Dept}_\textsf{B}}
\def\Staff{\textsf{Staff}}
\def\corecoverc{\textsf{CoreCover$\C$}}

\begin{definition} {\bf(Maximally Contained Rewriting (MCR))}
\label{def-mcr146}
Given a query $Q$ and a set of views $\V$, $R_{MCR}$ is a {\em maximally contained rewriting (MCR)} of $Q$ in terms of $\V$ if $R_{MCR}$ is a contained rewriting of $Q$ in terms of $\V$ and every contained rewriting $R$ of $Q$ in terms of $\V$ is contained in $R_{MCR}$.
\end{definition}

\begin{definition}{\bf(Certain Answers)}
Let $\V$ be a set of views and $\I$ be a view instance.
	Suppose there exists a database instance $D$ such that $\I\subseteq \V(D)$. Then,
	we define the certain answers of ($Q,\I$) with respect to $\V$ under  the Open World Assumption as follows:
		\[\text{certain}(Q,\I)=\bigcap\{Q(D): D \text{ such that } \I\subseteq\V(D)\}\] 

	If there is no database instance $D$ such that $\I\subseteq \V(D)$, we say that the set $\text{certain}(Q,\I)$ is undefined. 
\end{definition}

\begin{theorem}  
Given a query and a set of view definitions, a contained rewriting computes only certain answers.
\end{theorem}

\begin{proof}
The proof is based on the observation that if a set $A$ is a subset of set $A_1$ and $A$ is a subset of set $A_2$,
then $A$ is a subset of $A_1\cap A_2$. To prove, suppose element $e$ is in $A_1$  and in $A$ but not in $A_2$. 
This is a contradiction because if $e$ is in $A$ and $A$ is a subset of $A_2$, $e$ should be in $A_2$ too. The sets that we refer to, are the sets in $\{Q(D): D \text{ such that } \I\subseteq\V(D)\}$ vs. $R(\I).$
\end{proof}

\subsection{Union of CQAC MCRs compute all certain answers for CQAC queries and views}
\label{sec-6}

%

In this subsection, we will prove that, for  CQAC views, a
maximally contained rewriting $\P$  with respect to union of CQACs (U-CQAC)  \footnote{In the literature, usually, by U-CQAC we define the class of finite unions of CQACs, in this section we assume that it may  also be  infinite.} of a CQAC query $Q$ 
computes the certain answers of $Q$ under the OWA, i.e., we prove the following theorem.
\footnote{This section extends the results in \cite{AfratiK10} for CQs.}

\begin{theorem}\label{owa_mcrs}
	Let $Q$ be a CQAC query, $\V$ a set of CQAC views. Suppose there exists an MCR $\R_{MCR}$ of $Q$ with respect to U-CQAC.
Let $\I$ be a view instance such that the set $\text{certain}(Q,\I)$ is defined. 
	Then, under the open world assumption, $\R_{MCR}$ computes all the certain answers of $Q$ on any view instance $\I$ 
that is: $\R_{MCR}(\I)=\text{certain}(Q,\I)$.
\end{theorem}

\begin{proof}
We have already proved that any contained rewriting computes only certain answers. In order to prove the theorem, we willl prove now that, for any tuple $t_0$ of $\I$ such that  $t_0\in \text{certain}(Q,\I)$, there is a contained rewriting $R$ such that $t_0 \in R(\I)$.
%
%

We
consider the conjunction of all the atoms in $\I$ after we have turned 
the constants in $\I$ to variables and we have added ACs among these variables that define the total order of the constants in $\I$. Thus, we turn $\I$ into a Boolean CQAC query $R_{\I}$.  We consider the expansion, $R_{\I}^{exp}$, of $R_{\I}$.\footnote{It is easy to see that this is possible if the set $\text{certain}(Q,\I)$ is defined.} 

Suppose  a tuple $t_0$ in $\I$ is in the certain answers of $Q$ on $\I$. Consider the canonical databases 
 $D_{\I}^{exp-i}, i=1,2,\dots$ of $R_{\I}^{exp}.$ Then $t_0\in  Q(D_{\I}^{exp-i}), i=1,2,\dots.$ 
This remark leads us to the fact that $R^{t_0}_{\I}$  is a contained rewriting, where $R^{t_0}_{\I}$  is the following
 $R^{t_0}_{\I}$ is  $R_{\I}$ with head variables the ones corresponding to $t_0$. It is easy to see
$t_0\in  
R^{t_0}_{\I}(\I)$.
\end{proof}

\section{Contained Rewritings   for CQAC classes. Preliminaries.}
\label{sec-mcr-more}
%
%
%
For CQAC queries and views, we define the expansion of a rewriting:

\begin{definition}
	The \emph{view-expansion},\footnote{In Section \ref{subsec-mcr-datalogAC}, we will need to differentiate between view-expansion and Datalog-expansion which we will define shortly, therefore, when confusion arises we use these prefixes.} $R^{exp}$, 
 of a rewriting $R$ defined in terms of views in $\V$,  is obtained from $R$ as follows. For each subgoal $v_i$ of $R$ and the corresponding view definition $V_i$ in $\V$, if $\mu_i$ is the mapping from the head of $V_i$ to $v_i$ we replace $v_i$ in $R$ with the body of $\mu_i(V_i)$. The non-distinguished variables in each view are replaced with
	fresh variables in $R^{exp}$.
\end{definition}

\begin{theorem}
Given a CQAC $Q$, a set of views CQAC $\V$ and a CQAC rewriting $R$ of $Q$ using $\V$, $R$ is a {\em contained rewriting} if and only if $R^{exp}$  is contained in $Q$.
\end{theorem}

\begin{proof}
The ``if'' direction. 
Let  $R^{exp}\sqsubseteq Q$. We will prove that for every $\I$ such that $\I \subseteq \V(D)$ the following is true: 

Let $\I$ be a viewset and $D$ a database instance such that $\I \subseteq \V(D)$. 
Since CQACs are monotone queries $R(\I) \subseteq R( \V(D))$. 
We have that $R^{exp}(D) \subseteq Q(D)$. 
We compute $R^{exp}(D)$ by using a homomorphism $\mu$  from
 $R^{exp}$ to $D$ that satisfies the ACs in $R^{exp}$. When we produce the  $R^{exp}$ from $R$, we use the view definitions and the head of the view definitions in $R^{exp}$  map via $\mu $ on $D$ and produce tuples in $\V(D)$. These tuples can be used
to produce the homomorphism  from $R$ to $\V(D)$, hence $R(\V(D))$ is a subset of $Q(D)$ and, since $R(\I) \subseteq R( \V(D))$. 
$R(\I) \subseteq Q(D)$. 

For the other direction, suppose $R$ is a contained rewriting, hence we have:  For each database instance $D$ such that
$\I \subseteq \V(D)$, $R(\I)$ is a subset of or equal to $Q(D)$. We  also have that $R(\I )\subseteq R(\V(D))$, because of monotonicity of CQACs. The homomorphism from $R$ to $\V(D)$ that produces $R(\V(D))$ can be used to find a homomorphism from 
$R^{exp}$ to $D$ that produces the same tuples in $R^{exp}(D)$ as in $R(\V(D))$. Hence, $R(\I )\subseteq R^{exp}(D)$.
\end{proof}

\subsection{Rectified Rewriting}
\label{subsec-rectified}
This is a motivating example"
\begin{example}
	\label{ex:export-nondist1}
	Consider query $Q$ and view $V_2$:
	\begin{center}
		\begin{tabular}{l l}
			$Q(A)$     & $~\hbox{\rm :-}~~p(A), A < 4.$\\
			$V_2(Y,Z)$ & $~\hbox{\rm :-}~~p(X), s(Y,Z), Y \leq X, X \leq Z.$\\
		\end{tabular}
	\end{center}
	The following rewriting is a contained rewriting of the query in terms of the view in the language CQAC:
	$$R(Y_1)~~\hbox{\rm :-}~~V_2(Y_1,Z_1),V_2(Y_2,Z_2), Z_1\leq Y_2, Y_1\geq Z_2, Y_1 < 4.$$
	
	\noindent
	Now consider the following contained rewriting:
	
	$R'(X) $ ~\hbox{\rm :-}~ $V_2(X,X), X < 4.$
	
	\noindent
	This rewriting uses only one copy of the view. We can show that $R$ is not contained in $R'$ and that
	$R'$ is not contained in $R$. However, let us consider the rectified version of $R$:
	$$R^{rect}(Y_1)~~\hbox{\rm :-}~~V_2(Y_1,Z_1),V_2(Y_2,Z_2), Z_1\leq Y_2, Y_1\geq Z_2, Y_1 < 4, Y_1\leq Z_1, Y_2\leq Z_2.$$

We can easily prove now that $R'$ and $R^{rect}$ are equivalent.
	
	In fact, we can easily see that we can produce an arbitrarily long  contained rewriting with arbitrarily many relational subgoals.
\end{example}
Thus, we define the notion of rectified rewriting.
\begin{definition}
Consider  a CQAC query and a set of CQAC views.
When we have a  rewriting $R$ the  variables in $R$ also satisfy some ACs that are in the closure of the ACs in the expansion of $R$. We include those ACs in the rewriting $R$ and produce $R'$, which we call the {\em AC-rectified rewriting of $R$}. 
\end{definition}

Thus, the expansions of $R$ and $R'$ are equivalent queries. Hence, we derive the following proposition:

\begin{proposition}
	Given  a set of CQAC views, a  rewriting $R$  and its rectified version $R'$, the following is true: For any
	view instance $\I$ such that there is a database instance $D$  for which  $\I\subseteq \V(D)$, we have that $R(\I)=R'(\I)$.
\end{proposition}

\begin{definition}
	We say that a  rewriting $R$ is {\em AC-contained} in a   rewriting $R_1$  if the AC-rectified rewriting $R'$ of $R$ is contained in $R_1$ as queries.
\end{definition}
From hereon,
when we refer to a  rewriting, we mean the AC-rectified version of it and when we say that a   rewriting is contained in another  rewriting we mean that it is AC-contained.

\section{Computing Certain Answers in PTIME for RSI1 Queries and CQAC views}
\label{sec-mcrs-cert-answe}
%


In the rest of the paper, we will finish proving the main result of the second part of this paper, which is  the following theorem:

\begin{theorem}
\label{thm-mainsec61certain}
	Given a query $Q$ which is CQAC-RSI1 (RSI1, for short) with closed ACs  and views $\V$ which are
	CQACs,  we can find all certain answers of $Q$ using $\V$ on a given view instance $\I$ in time polynomial on the size of $\I$.
\end{theorem}
The above theorem is a straightforward consequence of  the main result of the previous  section and Theorem 	\ref{thm-mainsec61MCR} which says thast we can find an MCR in the language of Datalog with arithmetic comparisons.

Our method is as follows:
We first establish a containment test which transforms the containing query into a Datalog program and the contained query into a CQ and show that containment is equivalent to testing containment among the transformed queries (Theorem~\ref{thm:main123}).
An algorithm already exists that finds an MCR of a Datalog query in terms of  CQ views. We show that we can transform this MCR to  an MCR  of the original query in terms of the original views. This MCR is in the language of Datalog with ACs.

%
%
%
%
%
%
%

The following proposition together with Proposition \ref{trick-pro} is the main reason the transformation works:

\begin{proposition}
\label{pro-lsi-rsi-any-closed}
Consider the following containment implication
\begin{equation}
a_1~\wedge~ a_2 ~\wedge~ \cdots ~\wedge~ a_n \Rightarrow b_1\vee b_2 \vee \cdots \vee b_m \\ 
\end{equation}
where the $a_i$s are any AC and the  $b_i$s are closed SIs. When this implication is in minimal form it has at most two ACs on the rhs.
\end{proposition}

\begin{proof}
According to  Lemma \ref{lemm-3clauses}(4), when we prove a closed AC is in the closure of $F$ using a chain, then there is at most one constant. Since this AC is an SI, this constant is in the end of the chain, hence only one SI appears in this chain.
\end{proof}
The following containment implication
shows that  Proposition \ref{pro-lsi-rsi-any-closed} is not true if we use open SIs on the rhs. 
\vspace*{-.3cm}
%
\vspace*{.2cm}
	{\footnotesize$$X\neq Y  \Rightarrow   ~~X>7 ~~\vee ~~Y>7 ~~\vee ~~X<7~~~\vee ~~Y<7.$$}

\subsection{Preliminaries on Datalog Queries}

We include definitions and known results about Datalog. 
A \textit{Datalog query} (a.k.a.  Datalog program) is a finite set of Datalog rules, where a \textit{rule} is a CQ whose predicates in the body could either refer to a base relation (\textit{extensional} database predicates, \textit{EDBs}), or to a head of a rule in the query (\textit{intensional} database predicates, \textit{IDBs}). Furthermore, there is a designated predicate, which is called \textit{query predicate}, and returns the result of the query.
The atom whose predicate is an EDB (resp. IDB) is called \textit{base atom} (resp. \textit{derived atom}). 

During the computation, we use an {\em instantiated rule}, which  is a rule where all the variables have been replaced by constants.  We say that a rule is {\em fired }
if there is an instantiation of this rule where all the atoms in the body of the rule are are already computed in previous rounds or is in the given database instance.

The evaluation of a Datalog query on a database instance is performed by firing the rules  until no more facts (i.e., ground head atoms) are added to the set of the derived atoms. The answer of a Datalog query on a database is the set of facts derived during the computation for the query predicate. 
A $Datalog^{AC}$  query allows in each rule also arithmetic comparisons (ACs) as subgoals, i.e., each rule  is a CQAC. The evaluation process remains the same, only now, the AC subgoals should be satisfied too. 

We say that we \emph{unfold a rule} if we replace an  IDB subgoal with the body of another rule that has this IDB predicate in its head.
A \textit{partial expansion} of a Datalog query is a conjunctive query that results from unfolding the rules one or more times; the partial expansion may contain IDB predicates. A \textit{Datalog-expansion} (or, simply {\em expansion}, if confusion does not arise) of a Datalog query is a partial expansion that contains only EDB predicates. Considering all the (infinitely many) expansions of a Datalog query we can prove that a Datalog query is equivalent to an infinite union of conjunctive queries.
An expansion of a $Datalog^{AC}$  query  is defined the same way as an expansion of a Datalog query, only now we carry the ACs in the body of each expansion we produce. Thus, in an analog way, a $Datalog^{AC}$  query  is equivalent to an infinite union of CQACs.

\section{A new Containment Test for RSI1 containing query and CQAC contained query}
\label{subsec-construct-Datalog}
Without loss of generality, we restrict attention to Boolean queries.
First, we describe the construction of a Datalog query from a given RSI1 query $Q_1$.  
We introduce the predicates that we use:\\
\noindent
{\bf  Predicate Definition  List:}
\vspace*{-.20cm}
\squishlist

\item [1.] One EDB predicate for each relation in $Q_1$ of the same arity. 

\item [2.] One binary EDB predicate, called $U$.

\item [3.] Four unary IDB predicates for  each constant, $c$, that appears in an AC  in $Q_1$, we call them  $I_{\theta
		c}$  and  $J_{\theta c}$, where $\theta$ is either $\leq$ or $\geq$.
\item [4.]	A Boolean  IDB predicate called  $Q_1^{Datalog}$, which is the query predicate.
\squishend
\vspace*{.20cm}

The recursive rules of the Datalog query depend only on the containing query, $Q_1$. The non-recursive  rules  take into account a finite set of {\em relevant semi-interval (SI) ACs}  and we also call them {\em dependent rules} or {\em link rules}. The relevant 
SI ACs are typically ACs that are in the closure of ACs in the contained query, but not necessarily\footnote{the fact that there are other options will be clear later when we discuss about constructing rewritings }.
Moreover, for convenience of reference, we  divide the recursive rules into three kinds:  {\em mapping rules, coupling rules,} and a single {\em query rule}, which is  fired only once.

We describe the construction of the recursive rules of the Datalog query:

{\bf Construction of the recursive rules:}
\vspace*{-.20cm}
\squishlist
	\item [1.]
	The {\em query rule} copies into its body all the relational subgoals of $Q_1$, and
	replaces each AC subgoal of $Q_1$, $X\theta c$, with the atom $I_{\theta c}(X)$.
	\item [2.]
	We
	construct one {\em mapping rule}   for each $I_{\theta c}(X)$ in the query rule, i.e., in total $M$ mapping rules, where $M$ is the number of ACs in $Q_1$. 
The mapping rule for $I_{\theta c}(X)$ has the same body as the body in the query rule except that the $I_{\theta c}(X)$ is deleted. Its head of this mapping rule is $J_{\theta c}(X)$.
	\item [3.]
	For every pair of constants $c_1 \leq  c_2$ used in $Q_1$, we
construct three {\em coupling rules}.

	Two coupling rules of the {\em first kind}: 
	\small{\begin{center}
		\begin{tabular}{l}
			$I_{\leq c_2}(X)~\symif~J_{\geq c_1}(X)$ \\
			$I_{\geq c_1}(X)~\symif~J_{\leq c_2}(X)$
		\end{tabular}
	\end{center}}

One coupling rule of the {\em second kind}:
		\small{$$I_{\leq c_2}(X)\symif~J_{\geq c_1}(Y),U(X,Y).$$}
%
%
	
\squishend

\vspace*{-.5cm}
{\bf Construction of the dependent or link rules}
\vspace*{-.20cm}
\squishlist
	\item
For each relevant SI AC,  $X\theta c$, we introduce a unary base predicates,  $U_{\theta c}$ and we add the following link rule:
$$ I_{\theta c_1}(X):-U_{\theta c}(X).$$
where $c_1$ is a constant in $Q_1$, for which $X\theta c\Rightarrow X\theta c_1$
\squishend

%
%



\begin{example}
	\label{ex-running-dat}
The following query $Q_1$ is an  RSI1 query:
	\begin{center}
		\begin{tabular}{lll}
			$Q_1({})$ & $\symif$ &$ e(X,Y),e(Y,Z),X\geq 5,Z\leq 8.$ \\
		\end{tabular}
\end{center}
	For the query $Q_1$, the construction we described yields the following recursive rules of the Datalog
	query $Q_1^{Datalog}$:
	\small{
	\begin{center}
		\begin{tabular}{l l l}
			$Q_1^{Datalog}()$ $\symif$       & $e(X,Y),e(Y,Z),I_{\geq 5}(X),$& \multirow{2}{*}{(query rule)}\\
			&$I_{\leq 8}(Z).$& \\
			$J_{\leq 8}(Z)$ $\symif$ & $e(X,Y),e(Y,Z),I_{\geq 5}(X).$           & (mapping rule)\\
			$J_{\geq 5}(X)$ $\symif$ & $e(X,Y),e(Y,Z),I_{\leq 8}(Z).$           & (mapping rule)\\
			$I_{\leq 8}(X)$ $\symif$ & $J_{\geq 5}(X).$                         & (coupling rule 1)\\
			$I_{\geq 5}(X)$ $\symif$ & $J_{\leq 8}(X).$                         & (coupling rule 1)\\
			$I_{\leq 8}(X)$ $\symif$ &$J_{\geq 5}(Y), U(X,Y)$ & (coupling rule 2)\\
			$I_{\geq 5}(X)$ $\symif$ &$J_{\leq 8}(Y), U(Y,X)$ & (coupling rule 2)\\
		\end{tabular}
	\end{center}}
	Intuitively, a coupling rule denotes that a formula $AC_1 \vee AC_2$ ( for two
	SI comparisons $AC_1=X\theta_1 c_1$ and $AC_2=Y\theta_2 c_2$) is either true (coupling rule 1) or it is implied by $X\leq Y$ (which is encoded by  the predicate $U(X,Y)$) (coupling rule 2). Thus, the first coupling rule in the above query says that
	$ X\leq 8  \vee X\geq 5$ is true and the second coupling rule says the same but refering to different
	$I$ and $J$-atoms. Moreover, the last coupling rule says that $X\leq Y\Rightarrow X\leq 8  \vee Y\geq 5$.
\end{example}

\noindent
{\bf Construction of CQ for Contained Query}
\label{subsec-contained-trans}
%

We now describe the construction of the contained query turned into a CQ.

We introduce new unary
EDBs, specifically two of them, by the names $U_{\geq c}$ and $U_{\leq c}$, for each constant $c$ in $Q_2$.
Let us now construct the CQ $Q_2^{CQ}$ from $Q_2$.
We initially copy the regular subgoals of $Q_2$,
and for each SI  $X_i\theta c_i$ in the closure of $\beta_2$ we add a
unary predicate subgoal $U_{\theta c_i}(X_i)$. Then, for each AC $X\leq Y$ in the closure of ACs in $Q_2$, we add the unary subgoal $U(X,Y)$ in the body of the rule.

\begin{example}
\label{ex-cqcq}
 Considering the CQAC $Q_2$ with the following definition:
\small{
	\begin{center}
	\begin{tabular}{ll}
		$Q_2()\symif$ & $e(A,B),e(B,C),e(C,D),e(D,E), A\geq 6, E\leq 7.$ \\
	\end{tabular}
\end{center}}

we construct the $Q_2^{CQ}$ whose definition is:
\small{
	\begin{center}
	\begin{tabular}{ll}
		$Q_2^{CQ}()\symif$ & $e(A,B),e(B,C),e(C,D),e(D,E),U_{\geq 6}(A),U_{\leq 7}(E).$ \\
	\end{tabular}
\end{center}}
\end{example}
Thus the dependent rules for our running example, query $Q_1$, and the above contained query $Q_2$ are:

	\begin{center}
	\begin{tabular}{l l l}
		$I_{\geq 5}(X)$ & $\symif~U_{\geq 6}(X).$ & (link rule)\\
		$I_{\leq 8}(X)$ & $\symif~U_{\leq 7}(X).$ & (link rule)\\
	\end{tabular}
\end{center}


Now, we have completed the description of the construction of both $Q_1^{Datalog}$ from $Q_1$ and
$Q_2^{CQ}$ from $Q_2$. We go back to our examples and put all together in the following example:

\begin{example}
	\label{ex-full}
	Our contained query is the one in Example  \ref{ex-cqcq}.
	Our containing query is the one  in Example \ref{ex-running-dat}. The transformation of the containing query is shown in Example 	\ref{ex-running-dat}.  The transformation of the contained query is shown in Example \ref{ex-cqcq}. To complete the Datalog query, we add
	the following link rules:

	
	\begin{center}
		\begin{tabular}{l l l}
			$I_{\geq 5}(X)$ & $\symif~U_{\geq 6}(X).$             &             (link rule)\\
			$I_{\leq 8}(X)$ & $\symif~U_{\leq 7}(X).$              &            (link rule)\\
		\end{tabular}
	\end{center}
	 One rule links the constant 6 from the ACs of  $Q_2$ to the constant 5 from the ACs of  $Q_1$. The other link rule links constants 7 and 8 from queries $Q_1$ and $Q_2$, respectively.
\end{example}

The constructions of the Datalog query and the CQ presented in Section \ref{subsec-construct-Datalog} 
are important because of the following theorem.
%
%
%
%
%

%

\begin{theorem}
	\label{thm:main123}
	Consider two conjunctive queries with arithmetic comparisons, $Q_1$ and $Q_2$  such that $Q_1$ is an RSI1 query with closed ACs.
	Let  $Q_1^{Datalog}$ be the  transformed Datalog query of $Q_ 1$. Let
	$Q_2^{CQ}$ be the  transformed CQ query of $Q_2$.
	Then,  $Q_2$ is contained in $Q_1$   if and only if   $Q_2^{CQ}$ is contained in $Q_1^{Datalog}$.
\end{theorem}

The proof of Theorem \ref{thm:main123} is in  \ref{prf:thm-main123p}. In the next section, we begin to discuss some of the  technicalities involved, as an introduction to the  proof.
 \ref{subsec-examples} offers  intuition on how the transformations work.

\section{Preliminary  Results and Intuition on the Proof of Theorem \ref{thm:main123}  
}
\label{subsec-simplefacts}

%

%

We observe that the Datalog program we construct from a RSI1 query has only unary IDB predicates and we conveniently refer to them using the symbol we used to name them as $I$ predicates and $J$ predicates. $J$ predicates appear in the head of mapping rules and in the body of coupling rules and $I$ predicates appear in the body of mapping rules and in the head of coupling rules. We conveniently say that they produce $I$ facts and $J$ facts.

In the proof of Theorem \ref{thm:main123}, we will apply the Datalog query  $Q_1^{Datalog}$ on the canonical database of the CQ query $Q_2^{CQ}$ constructed from the contained query $Q_2$.
This canonical database uses constants (different from the constants in the ACs) that correspond one-to-one to variables of the query $Q_2$.

\squishlist
\item
We refer to the fact 
 $I_{\theta
		c}(x)$ (atom  $J_{\theta c}(x)$, respectively) as the {\em associated} fact of SI AC $X\theta c$ and vice versa.

\squishend
We do the following observations about the result of firing a recursive rule (i.e., either a coupling rule or a  mapping rule):
(All the $\theta _i$s represent either $\leq$ or $\geq$ and the $c_i$s are constants from the ACs of the queries.)

\squishlist
	\item {\em Firing coupling rules.} We have two kinds of coupling rules. Consider a coupling rule of the first kind which is of the form: 
$$I_{\theta_1 c_1}(X) \symif~J_{\theta c_2}(X).$$
	When this rule is fired, its variable $X$ is instantiated to a constant, $y$, in the canonical database, $D$, of $Q_{20}$. The constant $y$ corresponds to the variable $Y$ of $Q_2$ by convention. Then the following is true by construction: 
$Y\theta_1 c_1 \vee Y\theta_2 c_2$, and, hence, the following is true:
	$\beta_2\Rightarrow Y\theta_1 c_1 \vee Y\theta_2 c_2.$
	Now consider the second kind of coupling rule, which is of the form:
	$$I_{\theta_1 c_1}(X) \symif~J_{\theta c_2}(Y),U(X,Y).$$
	By construction of the rule, the EDB $U(X,Y)$    is mapped in $D$ to two constants/variables, $W,Z$, 
	such that there in $Q_2$ an AC which is $W\leq Z$.       Thus, by construction of the rule, the following is true again:
	$\beta_2\Rightarrow W\theta_1 c_1 \vee Z\theta_2 c_2.$
	\item {\em Firing both mapping and coupling rules.}
\begin{example}
For a first example, suppose only the query rule and link rules are needed to prove containment of $Q_2^{CQ}$ in $Q_1^{Datalog}.$
Suppose we fire a link rule to compute the fact $I_{\theta c}(x)$. This yields that $\beta _2 \Rightarrow X\theta c$ (by construction of the rule). Suppose that  after applying some link rules, we are able to fire the query rule using a mapping $\mu_1$. I.e., 
for each 
$\mu _{1}(e^{\beta_1}_{i}) $ (where $e^{\beta_1}_{i}$ is such that $\beta_ 1=e^{\beta_1}_{1} \wedge e^{\beta_1}_{2}\wedge \ldots$), we have $\beta_2  \Rightarrow \mu _{1}(e^{\beta_1}_{i})$. Consequently the following is true: $\beta_2  \Rightarrow \mu_1(\beta_1). $ This is a containment entailment which reassures that $Q_2\sqsubseteq Q_1$.

\end{example}

In the following,  fact
$I_{\theta c}(x)$ will be conveniently referred to as fact  $I(x) $  associated with  $Y\theta c$, 
Thus, we have proven the following lemma:

\begin{lemma}
\label{lem-coupling-implication}
 Suppose a mapping rule is fired using mapping $\mu$ and using the $I$ facts $I(x_1),\ldots,I(x_M)$, where each
$I(x_i)$ is associated with AC $\mu(e^{\beta_1}_{i})$ (without loss of generality). 
Suppose $I(x_i)$ is computed either via a link rule or via a coupling rule which used $J$ fact $J(x_{k_i})$. Suppose     fact $J(x_{k_i})$ is computed via a mapping rule  $\mu_{m_i}$ and hence $J(x_{k_i})$ is associated with an
AC $\mu_{m_i}(e^{\beta_1}_{l_i})$. Then the following is true: $\beta_2 \Rightarrow \mu(e^{\beta_1}_{i}) \vee  \mu_{m_i}(e^{\beta_1}_{l_i})$.

\end{lemma}

\squishend

\section{Finding MCR for CQAC-RSI1 Query and CQAC Views}
\label{subsec-mcr-datalogAC}

In this section, we show  that for a RSI1 query and CQAC views, we can find an MCR in the language of (possibly infinite)  union of CQACs.
We will show that this MCR is expressed in
Datalog with ACs. I.e., we prove:

\begin{theorem}
	\label{thm-mainsec61MCR}
	Given a query $Q$ which is CQAC-RSI1 (RSI1 for short) with closed ACs and CQAC views $\V$, we can find
an MCR of $Q$ using $\V$ in the language of Datalog with ACs.
\end{theorem}

%
%


\subsection{Building MCRs for RSI1 queries}
\label{subsec-buildingMCRs}
In this subsection, we present the algorithm for building an MCR in the language of (possibly infinite) union of CQACs for the case of CQAC views and queries that are RSI1, which is the following  (see also Figure \ref{fig:proof-mcr}):

%
%
%
%

\begin{figure}
		\centering
		\includegraphics[width=0.83\linewidth]{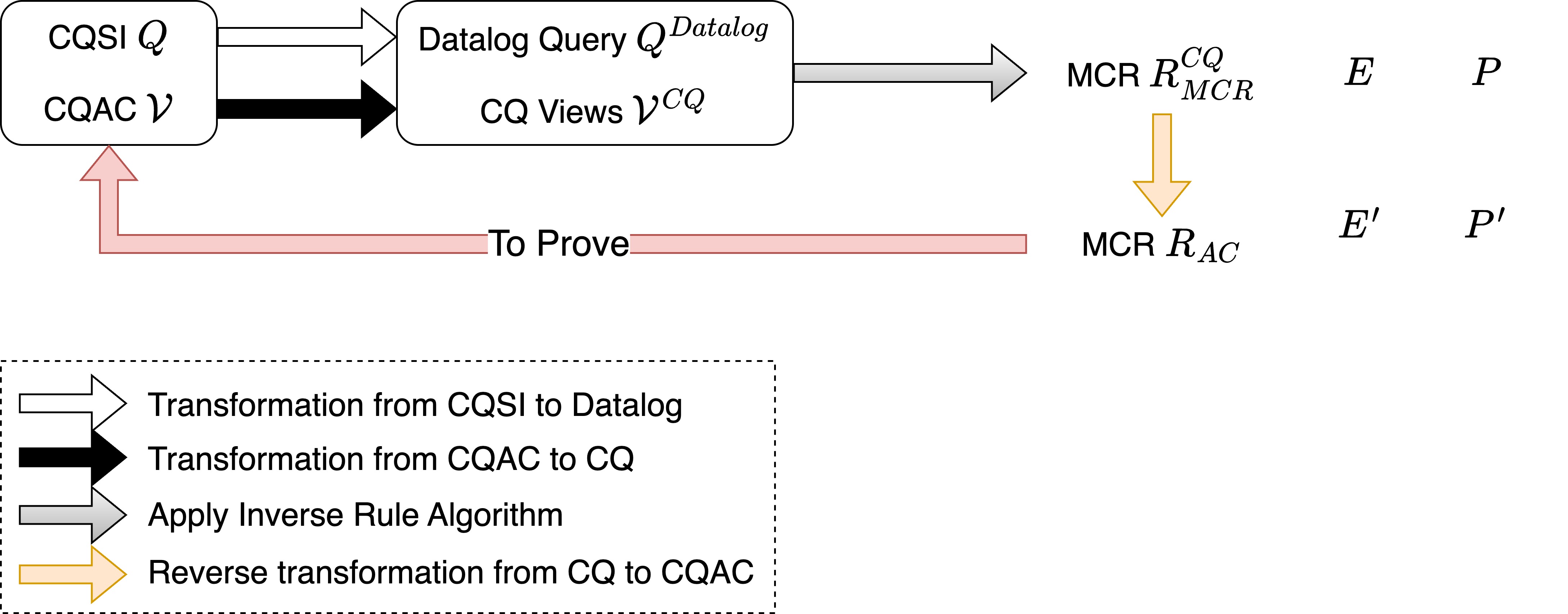}
		\caption{Finding an MCR}
		\label{fig:proof-mcr}
	\end{figure}


{\bf Algorithm MCR-RSI1}:

{\sl Stage I. Building CQ views and Datalog query}

\squishlist

	\item [1.] 
 For each view $v_i$ in $\V$, we construct a new view $V_i^{CQ}$ in $\V{^{CQ}}$ using the CQ transformation described in Section \ref{subsec-contained-trans}.
	
	\item [2.] In the view set $\V{^{CQ}}$we also add {\em auxiliary views}. We introduce new EDB predicates called {\em AC-predicates} (which will encode ACs $var \leq var$ and $var ~\theta ~const$) as follows:  a binary predicate $U$ and, a set of unary predicates $U_{\leq  c}$ and $U_{\geq  c}$, two for each constant $c$ that appears in the query and the views.  We construct the auxiliary views 
as follows: a) Views with head  $u_{\theta c}$, one for each
	unary predicate $U_{\theta c}$. The definition is $u_{\theta c}(X)~\symif~
	U_{\theta c}(X)$. b) A single view  $u$, whose definition is $u(X,Y)~\symif~
	U(X,Y)$. 

	\item [3.] For the query $Q$, we construct the Datalog query $Q^{Datalog}$
 using the construction in Section \ref{subsec-construct-Datalog}. The link rules will use the constants present in the views and in the query. The only difference is that we use IDB predicate $U^{tr}$ instead of EDB predicate $U$. $U^{tr}$ encodes transitive closure of the  $\leq$.

\item [4.] We finish the construction of the Datalog query by adding transitive closure rules for the transitive closure of the binary EDB predicate
$U$,  i.e., it is computed by the rules:  $U^{tr}(X,Y) \symif~U(X,Y)$
and $U^{tr}(X,Y) \symif~U(X,Z),U^{tr}(Z,Y). $ Moreover, for  appropriate  pairs of constants $c_1,c_2$ we add  rules of the following type:  $J_{\theta_1 c_1}(X) \symif~U_{\theta c_2}(Y),U^{tr}(X,Y).$

Thus, in the first stage,  from the original set $\V$ of views and  query $Q$, we build a set $\V{^{CQ}}$ of views and query $Q^{Datalog}$. This is illustrated by the two top  boxes in Figure \ref{fig:proof-mcr} and the arrows indicate the transformations.

{\sl Stage II. Building MCR}

	
	\item [5.] We find an MCR, $R^{CQ}_{MCR}$, for the Datalog query $Q^{Datalog}$ using the views
	in $\V{^{CQ}}$. For building $R^{CQ}_{MCR}$ we use the inverse rule algorithm  \cite{Duschka97-3}.
	
	\item [6.] We obtain rewriting $R_{AC}$ as follows: we replace in the found MCR $R^{CQ}_{MCR}$, each $v_i^{CQ}$ by
	$v_i$, each $u_{\theta c}(X)$ by arithmetic comparison $X \theta c$ and each $u(X,Y)$ by arithmetic comparison $X\leq Y$. This is what we call reverse transformation in Figure \ref{fig:proof-mcr}..
\squishend
We are called to prove that the rewriting $R_{AC}$ is an MCR of $Q$ using 
$\V$.
In the proof we will use the rewritings $E,E'$ and $P,P'$ which are added in the far right end of Figure \ref{fig:proof-mcr} to guide the reader during the proof in Subsection \ref{subsec-proof-mcr1}.

\subsection{Proof that the algorithm {\bf MCR-RSI1} is correct}
\label{subsec-proof-mcr1}
We need some preliminary results.
We need to mention that  adding the transitive closure on $U$ in $Q^{Datalog}$ does not affect the result in 
Theorem 	\ref{thm:main123}.


\begin{lemma}
	\label{lemm-transitive}
	Consider two conjunctive queries with arithmetic comparisons, $Q_1$ and $Q_2$  such that $Q_1$ is an RSI1 query with closed ACs.
	Let  $Q_1^{Datalog^{tr}}$ be the  transformed Datalog query of $Q_ 1$ enhanced with transitive closure on the special EDB predicate $U$. Let
	$Q_2^{CQ}$ be the  transformed CQ query of $Q_2$.
	Then,  $Q_2$ is contained in $Q_1$   if and only if   $Q_2^{CQ}$ is contained in $Q_1^{Datalog^{tr}}$.
\end{lemma}

\begin{proof}
EDB relation $U$ encodes $\leq$ or $\geq$ of a logically closed set of ACs. Thus, if we compute transitive closure of $U$ we obtain again $U$. Hence
$U $ and $U^{tr}$ is the same relation.
\end{proof}
The following example shows that for our argument about the algorithm building indeed an MCR we need to prove that it builds an MCR with rectified rewritings.

\begin{example}
Recall  Example \ref{ex:export-nondist1}.
According to Lemma \ref{lemma-direct}(2) the homomorphism property holds. For finding an MCR we can apply the  algorithm in this paper, only that
in this case the Datalog transformation will produce a CQAC. The MCR found this way will produce  $R'(X)$ in 
Example \ref{ex:export-nondist1} which is a 
a rectified rewriting.  Rewriting $R'(Y_1)$ in 
Example \ref{ex:export-nondist1}  is also an MCR and there are arbitrarily many such MCRs with arbitrarily many relational subgoals. 
%
\end{example}

\begin{lemma}
The MCR found by the algorithm is such that the Datalog program that constitutes the MCR contains Datalog-expansions that are CQACs that are contained rewritings that are rectified.
\end{lemma}

\begin{proof}
The MCR  $R_{AC}$ results from an MCR $R^{CQ}_{MCR}.$ The latter only contains explicit  predicates that encode transitive closure of $U$. These predicates are translated into ACs in $R_{AC}$, thus all the ACs are explicit in 
$R_{AC}$.
\end{proof}

We consider a Datalog expansion $E'$ of $R_{AC}$, which corresponds to a Datalog expansion, $E$, of $R^{CQ}_{MCR}$ (by construction)\footnote{We put $E$ and $E'$ in the far end of Figure \ref{fig:proof-mcr} as a reminder}. 
We consider the view-expansions of $E$ and $E'$,  $E^{exp}$ and $E'^{exp}$ respectively.
$E^{exp} $ is the CQ transformation of $E'^{exp}$ (by construction of $R_{AC}$).
$E$ is a contained rewriting of $Q^{Datalog}$, hence $E^{exp} $ is contained in $Q^{Datalog}$.
According to the containment test based on transformations (Datalog transformation and CQ transformation),  $E'^{exp}$ is contained in the query $Q$, hence $E'$ is a contained rewriting of $Q$. This proves that MCR $R_{AC}$ is a contained rewriting of the query $Q$. It remains to prove that it is maximally contained.

Consider any contained CQAC rectified rewriting, $P'$, of $Q$ using the view in $\V{^{CQ}}$. Then, $P'^{exp}\sqsubseteq Q$, and, according to the containment test via the transformations, \\$P'^{exp-CQ}\sqsubseteq Q^{Datalog}$,
where $P'^{exp-CQ}$ is the CQ transformation of $P'^{exp}$. 

We construct  $P$ to be a query in terms of  the views
	in $\V{^{CQ}}$ which
uses the relational body of $P'$ and all the AC-EDB predicates of $P'^{exp-CQ}$ that use only variables in $P'$. 
(By AC-EDB predicates we refer to the EDB pedicates $U$ and $U_{\theta c}$ that appear in the Datalog transformation  encoding ACs.)
The view-expansion of 
$P$, $P^{exp}$, is equal to $P'^{exp-CQ}$ by construction;  hence,  $P^{exp}\sqsubseteq Q^{Datalog}$. 

Consequently, $P$ is a contained rewriting of $Q^{Datalog}$ in terms of $\V'$; hence, it is contained in the MCR, i.e.,  $P\sqsubseteq R^{CQ}_{MCR}$. 

Consequently, there is a 
Datalog expansion, $E$,  of $R^{CQ}_{MCR}$ from which there is a containment mapping on $P$.  If we apply reverse transformation on $E$ (i.e., replace the AC-EDB subgoals with ACs) then we get a Datalog-expansion, $E'$, of the MCR $R_{AC}$. Notice that $P$ and $P'$ are $``$isomorphic$"$ if we replace the AC-EDB with ACs and vice-versa
($P$ is a rectified rewriting, so all the ACs in the view-expansion involving  variables of $P$ appear in $P$).
Hence, there is a containment mapping from $E'$ to $P'$.,  Hence $P'\sqsubseteq  E'$, and, consequently  
$P'\sqsubseteq  R_{AC}$.
The table below shows pairs (vertically), each pair being a  CQAC or a Datalog$^{AC}$ query  and the corresponding CQ transformation. Notice that we assume that $E'$ and $P'$ contain all ACs in the closure of ACs in the rewriting, while $E$ and $P$ only the present AC-EDB predicates. However, the Datalog program discovers 
all of them by using the transitive closure rules on the AC-EDB binary predicate $U$ that encodes $\leq$.

  \small{ \begin{center} \begin{tabular} {|c|c|l|c|c|} \hline
%
%
%
%
%
%
%
%

$P'$& $P'^{exp}$  & $E'$&   $R_{AC}$ &{\bf with ACs}    \\\hline
 $P$&   $P^{exp}= P'^{exp-CQ}$  & $E$ &    $R^{CQ}_{MCR}$ &{\bf with AC-EDB predicates}         \\\hline

\end{tabular}

\end{center} }  

Thus, we have proven the theorem:

\begin{theorem}
	\label{thm-mcr-cont}
Given a query $Q$ which is RSI1 and views $\V$ which are
	CQACs, the following is true:
	Let $R$ be a  CQAC contained rewriting of $Q$ in terms of $\V$. Then $R$ is contained in the one found by the algorithm in Subsection~\ref{subsec-buildingMCRs} Datalog$^{AC}$ program $R_{MCR}$.
\end{theorem}

The following example is a simple one, on which it is easy to see that the found MCR by the algorithm is a rectified 
rewriting.

Theorems  	\ref{thm-mainsec61certain}   and  	\ref{thm-mainsec61MCR} can be extended, in a similar manner as in the case of deciding the complexity of 
query containment to the following:

\begin{theorem}
\label{thm-mainsec61certainMCR}
	Given CQAC views $\V$ and a query which is one of the following:

(i) It uses closed LSIs and one open RSI and the constant in the RSI is not shared with a closed RSI in any view definition.

(ii)  It uses open LSIs and one closed RSI and the LSIs use all distinct  constants and each such constant is not shared with a closed LSI in any view definition.

(iii)  It uses open LSIs and one open RSI and all the SIs use distinct constants and such constant is not shared with a closed SI in any view definition.

Then, the following are true:

1. We can find
an MCR of $Q$ using $\V$ in the language of Datalog with ACs.

	2. We can find all certain answers of $Q$ using $\V$ on a given view instance $\I$ in time polynomial on the size of $\I$.
\end{theorem}

In \ref{sec-extend-single-mapp}, we extend the results on finding MCRs to include single mapping variables as per
Section \ref{sec-exte-single-mappinga-var}.

\section{Conclusions}
In this paper we have investigated  the computational complexity of query containment for CQACs and of computing certain answers in the framework of answering CQAC queries using CQAC views. 
We begin by looking into cases where the containing query uses only LSI ACs.
When the containing query uses only closed LSIs, the problem is in NP.  When there are open LSIs in the containing query it is not the case. In that respect, when the containing query uses only LSI and certain constants do not appear in the contained query, the problem is in NP (even more interestingly, via the homomorphism property). However, if the containing query uses open LSI and certain constants are allowed to appear in both queries, the problem becomes 
$\Pi^p_2$-complete.  Thus, we have delineated a boundary between NP and   $\Pi^p_2$-complete which surprisingly puts ''very similar'' problems in different computational classes. 
Then, we are investigating cases where the containment problem is in NP when the containing query uses both LSI and RSI ACs. This needs a more complicated algorithm to prove. 

Open problems remain when the containing query  uses both LSI and RSI ACs.  The most tight open problem is the complexity of query containment when the containing query 
uses two closed LSI ACs and two closed RSI ACs. The technique used to prove $\Pi^p_2$ hardness here does not work because the containing queries used in the proof use a number of ACs that is proportional to the size of the formula from which we do the reduction. 
%

Towards investigating similar problems, we believe that, 
if the relational subgoals of the containing query form an acyclic hypergraph and there are only several closed LSIs and one closed RSI, then it is worth investigating whether 
testing containment may be done in polynomial time. We already know that the CQ query containment problem is polynomial  when the containing query is acyclic.

The second part of the present paper considers finding MCRs and computing certain answers. 
First, we present a  result which says that, for CQAC query and views, an MCR in the language of union of CQACs computes exactly the set of certain answers. 
Containment tests usually provide the basis of algorithms that find MCRs. We use a containment test via transformations.
We show that in the case the query has only LSI and a single RSI, there is an MCR in the language of Datalog with ACs, we, consequently, show that we can compute certain answers in polynomial time for this case. As concerns broader classes of queries that contain any number of LSI and RSI ACs, we believe it is unlikely that there is an MCR in the language of union of CQACs, hence, probably the problem of computing certain answers is harder than a PTIME problem.
For MCRs, when the homomorphism property holds, various efficient techniques like the Minicon algorithm \cite{Pottinger00-1} may find an MCR in the language of union of CQAC. As for the corner cases of query that contains only LSI ACs and containment is proven to be $\Pi^p_2$-complete, the problem of finding an MCR in the language of union of CQACs (and the problem of computing certain answers in polynomial time) is an interesting open problem.
On a similar line of research that concerns equivalent rewritings, 
when the query is an acyclic CQ, a recent result  \cite{GeckKSS23} shows that there is an equivalent rewriting which is acyclic, if there is one at all. Probably it is worth investigating the problem with ACs, starting with simple cases, e.g., 
when the query contains only SIs or even only LSIs.

In conclusion, for the problem of query containment with SI ACs,  we built an interesting picture which is depicted in Table \ref{fig:table-results}, and for the problem  of computing certain answers of queries with SI ACs, we made progress in the direction of computing certain answers in polynomial time. 

{\footnotesize
\bibliographystyle{alpha}
\bibliography{references
}}
\appendix

\section{More on Reasoning about arithmetic comparisons}
\label{app-A-A}

\subsection{Proof of Lemma \ref{lemm-3clauses}}

\begin{proof}
For the first clause: Suppose the first time implication 8  is used, we have in $F$ the following ACs $X\leq Z, Z\leq Y ,
X\leq W,  W\leq Y, W\neq Z$. The second time it is used, using the result of the first time (which is $X\neq Y$), 
suppose we have $X_1\leq Y, Y\leq Y_1 ,
X_1\leq X,  X\leq Y_1 , X\neq Y$  which  derives $X_1\neq Y_1$. This (i.e., the $X_1\neq Y_1$) could have been derived using  the elemental implication 8 only once as follows: 
 $X_1\leq Z, Z\leq Y_1 ,
X_1\leq W,  W\leq Y_1 , W\neq Z$.  The  $X_1\leq Z$ is derived from the $X_1\leq X$ and $X\leq Z$. Similarly, the other three ACs are derived.  $Z\leq Y_1 $ is derived from the $Z\leq Y$ and $Y\leq Y_1$.
$X_1\leq W$ is derived from the $ X_1\leq X$   $ X\leq W$. $W\leq Y_1$  is derived from the $ W\leq Y$   $Y\leq Y_1 $.

For the second clause, the proof is obvious, since, only implications (2) and (7) derive $X\leq Y$ and the chain is an obvious consequence of applying implication (7) several times.

For the third clause, we observe that $X< Y$ can be derived either from elemental implications (4) or (6). 
If only (6) is applied several times then we derive  a chain. If (4) is used, then we need to show first $X\neq Y$. In order to show $X\neq Y$, according to clause (a) of the present lemma, we need to use (8) only once, and this means that  there is a chain from $X$ to $Y$ containing $W$ and 
a chain from $X$ to $Y$ containing $Z$, with $W\neq Z$ in $F$.
%

For the fourth clause, just notice that if a chain has more than two constants, then we can choose the two constants, one closest to the beginning of the chain and one closest to the end of chain and we form the chain by using  the arithmetic comparison between these two constants.
\end{proof}

\subsection{Analyzing the containment implication}

\begin{lemma}
\label{lem-orac}
Suppoe the containment implication (1) is true and is  in minimal form. Then for any AC, say $b_m$,  on the rhs, the following is true:
$$a_1~\wedge~ a_2 ~\wedge~ \cdots ~\wedge~ a_n \wedge  \neg b_1~\wedge~ \neg b_2 ~\wedge~ \cdots  \Rightarrow b_m$$
\end{lemma}

Lemma \ref{lem-orac} is a consequence of the following general lemma:
\begin{lemma}
If $p\Rightarrow q\vee r$ and $p\not\Rightarrow r$, then $p\wedge \neg r$ is satisfiable and  $p\wedge \neg r \Rightarrow q.$
\end{lemma}

\begin{proof}
Assume $p\Rightarrow q\vee r$ and $p\not\Rightarrow r$. Since $p\not\Rightarrow r$, there exists a model  for 
$p\wedge \neg r$, so $p\wedge \neg r$ is satisfiable.  To show that $p\wedge \neg r \Rightarrow q$ let $N$ be a model for $p\wedge \neg r $. Then $N$ is a model for $p$ and is not a model for $r$. From $p\Rightarrow q\vee r$, it follows that $N$ is a model for  $q\vee r$. Since $N$ is not a model for $r$, it follows that it is a model for 
$q$.
\end{proof}

%
We will use Lemma \ref{lem-orac} often, and, for convenience, we will say that   ''we move the $b_1, \ldots,b_{m-1}$  to the left hand side (lhs for short) and we apply Lemma \ref{lem-orac} to prove $b_m$.''
A consequence of Lemma \ref{lem-closure} 
 is the following:
\begin{lemma}
\label{lem-cccppp}
We can check whether a containment implication is true in polynomial time.
\end{lemma}
\begin{proof}
We move all the ACs  to the left hand side and compute the closure of these ACs. If there is an AC $a$ such that both $a$  and $\neg a$ are in  the closure   then the containment implication is true. 
Otherwise, not.
\end{proof}

%

Now we begin to focus on semi-interval ACs. 
It is trivial to see that, under our assumptions, if a result is valid for only LSI, then it is valid for only RSI as well, or, if a result is valid for two LSI and an arbitrary number of RSI in the query, then it is valid for two RSI and 
an arbitrary number of RSI. Thus, considering  a result, we define its symmetrical result to be the result where we 
replace LSI with RSI and RSI with LSI. Thus,
from hereon, and for the rest of the paper, we will state the results only for one of the symmetrical variants.

The following lemma roughly says that if we have only LSIs on the rhs of a containment implication which is in minimal form, then there is only one disjunct on the rhs. However, there is an exception in the case there is a ``$\neq$'' on the lhs which is stated in detail in the following lemma:
\begin{lemma}
\label{lemma-direct}
Consider the containment implication ( \ref{aa}), where the $a_i$s are from a set of ACs $A$ and the $b_i$s are from a set of ACs $B$. Then, for the  pairs of $A$ and $B$ listed below the following is true: If the containment implication ( \ref{aa}) is true and is in minimal form, then  the rhs has one disjunct.

\squishlist
\item [1.]
$A$ is a set of ACs of AC-type  $T_{AC}$, $B$ is a set of ACs of AC-type $\{var\! \leq \!const
 \}$.



\item [2.]
$A$ is a set of ACs of AC-type  $T_{AC}\setminus \{ var \neq var, var \neq const \}$, $B$ is a set of ACs of AC-type $\{var\! \leq \!const, var\! < \!const
 \}$.

\item [3.]
$A$ is a set of ACs of AC-type  $T_{AC}$, $B$ is a set of ACs of AC-type $\{var\! \leq \!const, var\! < \!const
 \}$.
and the following condition is satisfied: 
For 
any $X\neq Y$ that appears in  $A$,     if a constant (say $c_0$) is related by an AC to both $X$ and $Y$  in $A$   then, either (i) $c_0$ does not relate to both $X$ and $Y$ by a closed AC in   $A$ or  (ii) $c_0$ does not appear in an  open AC in   some $B$. 

\squishend
\end{lemma}


%

%



\begin{proof}
For the first case, 
we apply Lemma \ref{lem-orac} to prove a closed LSI, $b_i $ (let it be  $X\leq c$).  When we move each  LSI except $b_i$  to the  lhs, this becomes an RSI, so (according to    Lemma \ref{lemm-3clauses}, clause 2) it will not be used to the derivation of  LSI $X\leq c$.
Hence, the AC $b_i$ is implied by applying the elemental implications on the ACs of the set $A$ only.
Therefore,
the containment implication is in minimal form only with one AC on the rhs.
The second case is similar with the first case, since the absence of ``$\neq$''  in $A$ results in an open LSI in $B$ being implied by using a chain (Lemma \ref{lemm-3clauses}, clause 3) as in the first case.

Third case: If there is a closed LSI, $b_i$,  among the ACs in set $B$, then, according to Lemma \ref{lem-orac}, we move all other ACs of the containment implication on the lhs, then the argument is the same as above, except the following fine point: An open LSI, when moves on the lhs, becomes a closed RSI which together with a closed
LSI among the ACs in the set $A$ produces equality, which may prove $b_i$. However, if so, and taking into account  Lemma \ref{lemm-3clauses} clause 4, AC $b_i$ can be already proven by using only ACs in the set $A$.
So, in this case, the containment implication in minimal form has one AC on the rhs.
If there are only  open ACs in the set $B$, then, according to Lemma \ref{lemm-3clauses} clause 3, there is the possibility of proving a certain $b_i$ by using one closed RSI resulting from a  $b_j$ moved to the lhs and one closed LSI from the set $A$ together with a inequality, and a closed LSI from $A$, i.e.,
say the inequality mentioned  in  Lemma \ref{lemm-3clauses} clause 3, is 
 the one chain  mentioned in   Lemma \ref{lemm-3clauses} clause 3, is $X\neq Y$, one chain is $X\leq c$ and the other is $Y=c$ which is resulting from $Y\leq c$ from $A$ and $Y\geq c$ moved from $B$ to the lhs. For this to happen we need two closed LSIs in $A$ and two open LSIs in $B$, all sharing the same constant. This is the condition we excluded in the statement of this lemma, hence, in this case too, the containment implication has only one AC in the rhs when in minimal form.

%
%
%
\end{proof}

%
%
%
%
%
%

If the condition mentioned in clause 3 of Lemma \ref{lemma-direct}  is violated, then the conclusion of the  lemma is not  guaranteed to be true. Here is a simple  example of a violation where the conclusion of the lemma is not true.:

\begin{example}
\label{ex-one-ac-constant}
Consider the following containment implication: 

$$X\neq Y ~\wedge~ 5\leq Y ~\wedge~ 5\leq X
\Rightarrow   
X>5 \vee Y>5$$
The above  containment implication is in minimal form but it has two ACs on the rhs.
\end{example}

The above implication appears in the following example: 

$Q_1() :- a(X,Y),  X>  5 $~~~~~$Q_2() :- a(X,X' ), a(Y,X''), X\neq Y,X\geq 5,Y\geq 5 $

\subsection{Analysing the containment entailment}
\label{sec-analysis}
%
Consider the containment entailment (as in Theorem \ref{thm:cont-CQAC}).
$$ \beta_2  \Rightarrow \mu_1(\beta_1) \vee \cdots \vee
\mu_k(\beta_1).$$

where $\beta_2 $ is the conjunction of ACs in the contained query and $\beta_1$ is the conjunction of ACs in the containing query. Applying the distributive law, this containment entailment can be equivalently viewed as  a collection of {\em containment implications}, each containment implication being:
$$\beta_2 \Rightarrow  a_1 \vee a_2  \vee \cdots  a_k$$ where $a_i$ is one of the ACs in the conjunction $ \mu_i(\beta_1)$, $i=1,2,\ldots, k$.

\begin{example}
	\label{ex-firstLSI}
	For an example, consider the following (normalized) CQACs.
\vspace*{-.5cm}
	%
	\begin{center}
		\begin{tabular}{lll}
			$Q_1: q()$ & $\symif$ &$a(X_1,Y_1,Z_1), X_1=Y_1, Z_1< 5$ \\
			$Q_2: q()$ & $\symif$ &$a(X,Y,Z'), a(X',Y',Z), X\leq 5, Y\leq X, Z\leq Y, $$ $$X'=Y', Z'<5$\\
		\end{tabular}
	\end{center}
	
	Testing the containment $Q_2
	\sqsubseteq  Q_1$, it is easy to see that there are the following two containment mappings:
\vspace*{-.4cm}
	\squishlist
		\item $\mu_1:  X_1\rightarrow X, Y_1\rightarrow Y, Z_1\rightarrow Z'$
		\item $\mu_2:  X_1\rightarrow X', Y_1\rightarrow Y', Z_1\rightarrow Z$
	\squishend
	
	Hence, the containment entailment is given as follows:
	
	\begin{center}
		\begin{tabular}{l}
			$X\leq 5  \wedge Y\leq X\wedge  Z\leq Y\wedge    X'=Y'\wedge  Z'<5\Rightarrow$\\
			$\big(\; \mu_1(X_1)\!\!=\!\!\mu_1(Y_1) \;\wedge\; \mu_1(Z_1)\!\!<\!5\;\big)\;\vee$ \\
			$\big(\;\mu_2(X_1)\!\!=\!\!\mu_2(Y_1) \;\wedge\; \mu_2(Z_1)\!\!< 5\;\big)$\\
		\end{tabular}
	\end{center}
	
	which is equivalently written:
	
	\begin{center}
		\begin{tabular}{l}
			$X\leq 5  \wedge Y\leq X\wedge  Z\leq Y\wedge  X'=Y'\wedge  Z'<5  \Rightarrow$\\
			$(X=Y \wedge Z'< 5)  \vee (X'=Y' \wedge Z< 5)$\\
		\end{tabular}
	\end{center}
	Now we consider the containment entailment we built above. According to
	what we analyzed in this section, we can equivalently rewrite this containment entailment  by
	transforming its right hand side into a conjunction, where each conjunct is a disjunction of ACs. The transformed entailment is the following, where $\beta$   is the conjunction  $X\!\leq \!5  \wedge Y\!\leq \!X\wedge  Z\!\leq \!Y\wedge  X'\!\!=\!\!Y'\wedge  Z'\!<\!5$:
\vspace*{-.5cm}
	\begin{center}
		\begin{tabular}{lll}
			$\beta \Rightarrow \; (X\!\!=\!\!Y\vee X'\!\!=\!\!Y')\wedge(X\!\!=\!\!Y\vee Z\!< \!5)\wedge(Z'\!< \! 5\vee X'\!=\! Y')\wedge(Z'\!< \! 5\vee Z\!<\! 5)$ \\
		\end{tabular}
	\end{center}
So, we have, in this case, four containment implications, one of which is, e.g., \\$\beta\Rightarrow X\!\!=\!\!Y\vee X'\!\!=\!\!Y.$

\begin{lemma}
\label{lem-forgott}
Consider a containment entailment, $E$. 
The following two are equivalent:

a) One disjunct in the rhs of $E$ suffices to make the containment entailment  $E$ true.

b) Any containment implication in minimal form produced by $E$   has one disjunct in the rhs.

\end{lemma}
\begin{proof}
(b)$\Rightarrow$(a):  Suppose (b) is true and (a) is not true. Then, consider, from each disjunct of the containment entailment,  the AC which is not implied by the lhs of the entailment. These ACs (taken over all disjuncts) build a containment implication for which clause (b)  in the lemma is not true, hence contradiction.
\end{proof}
	
	%
	%
	%
	%
	%
	%
	%
	%
\end{example}

\section{On the transformation of CQAC queries to Datalog and CQ queries. The tree-like structure of the containment entailent for CQASI1 containing query}
\label{subsec-examples}

This appendix offers intuition on the proof of Theorem~\ref{thm:main123}.
Proposition
\ref{trick-pro}  begins to show a tree-like structure of the containment entailment and it gives the first intuition for constructing a Datalog query from the containing query that will help in deciding query containment. The following example gives an illustration of this intuition.

\begin{example}
	\label{ex:description-of-prop52}
	Let us consider the following two Boolean queries.
	\begin{center}
		\begin{tabular}{lll}
			$Q_1: q()$ & $\symif$ &$a(X,Y,Z),X\leq 8,Y\leq 7,Z\geq 6.$ \\
			$Q_2: q()$ & $\symif$ &$a(X,Y,Z),a(U_1,U_2,X),a(V_1,V_2,Y),$\\
			& &$ a(Z,Z_1,Z_2),a(U_1',U_2',U_1),a(V_1',V_2',V_1),$\\
			& &$ U_1'\leq 8,U_2'\leq 7,U_2\leq 7,V_1'\leq 8,$\\
			& &$ V_2'\leq 7,V_2\leq 7,Z_1\leq 7,Z_2\geq 6.$\\
		\end{tabular}
	\end{center}
	
	\begin{figure}
		\centering
		\includegraphics[width=0.30\linewidth]{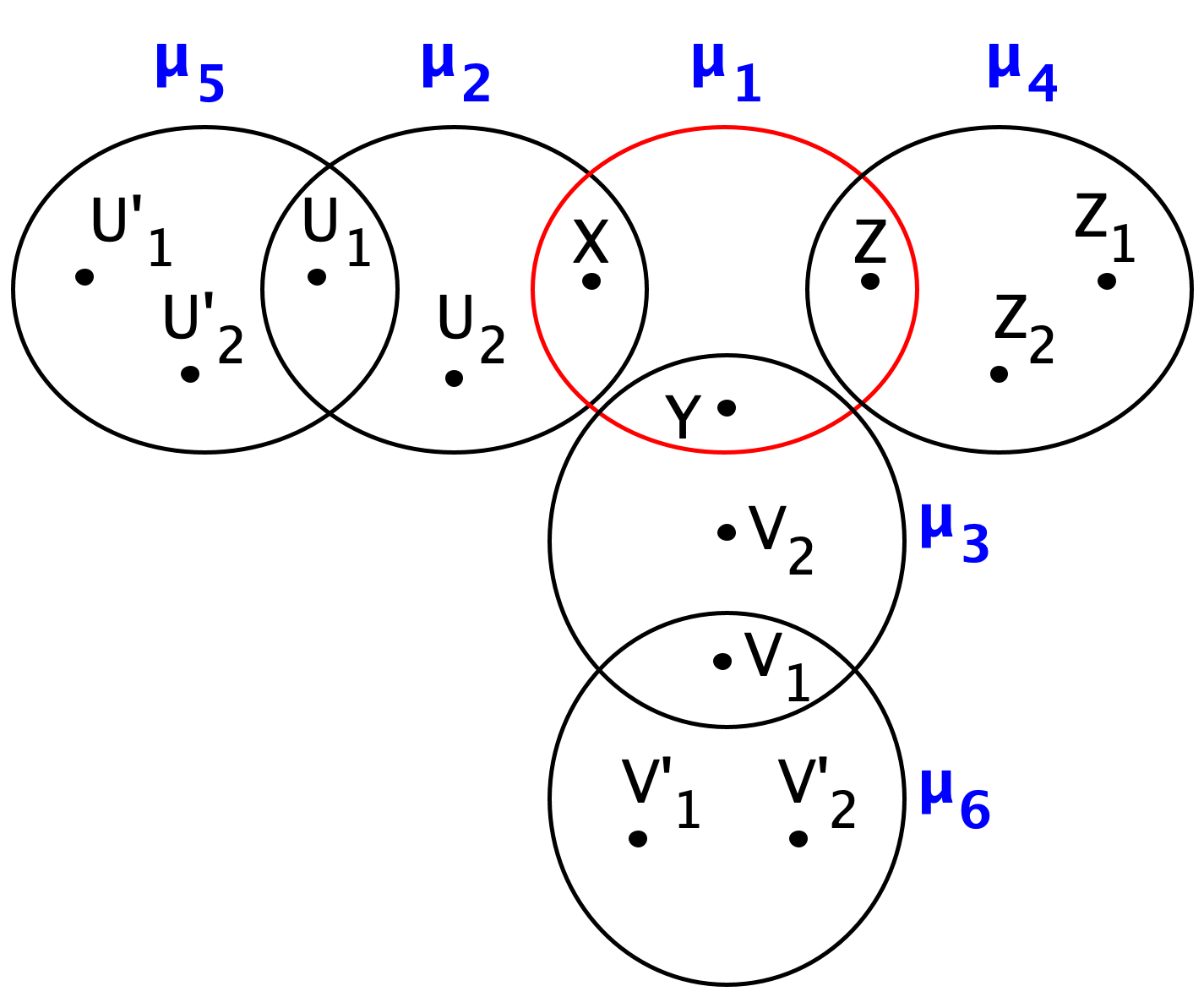}
		\caption{Illustration of containment entailment of Example~\ref{ex:description-of-prop52}}
		\label{fig:intuition-bodyq2}
	\end{figure}

	The query $Q_2$ is contained in the query $Q_1$. To verify this, notice that there are $6$ containment mappings from
$Q_1$ to $Q_2$. These mappings are given as follows: $\mu_1:$ $(X,Y,Z)\rightarrow(X,Y,Z)$, $\mu_2:$ $(X,Y,Z)\rightarrow(U_1,U_2,X)$, $\mu_3:$ $(X,Y,Z)\rightarrow(V_1,V_2,Y)$, $\mu_4:$ $(X,Y,Z)\rightarrow(Z,Z_1,Z_2)$, $\mu_5:$ $(X,Y,Z)\rightarrow(U_1',U_2',U_1)$, and $\mu_6:$ $(X,Y,Z)\rightarrow(V_1',V_2',V_1)$. After replacing the variables as specified by the containment mappings, the query entailment is $\beta\Rightarrow\beta_1\vee\beta_2\vee\beta_3\vee\beta_4\vee\beta_5\vee\beta_6$, where:
	
	\begin{center}
		\begin{tabular}{ll}
			\multicolumn{2}{l}{$\beta:$ $U_1'\leq 8\;\wedge U_2'\leq 7\;\wedge U_2\leq 7\;\wedge V_1'\leq 8\;\wedge V_2'\leq 7\;\wedge V_2\leq 7\;\wedge Z_1\leq 7\;\wedge Z_2\geq 6$.}\\
			$\beta_1:$ $X\leq 8\;\wedge Y\leq 7\;\wedge Z\geq 6$. & $\beta_4:$ $Z\leq 8\;\wedge Z_1\leq 7\;\wedge Z_2\geq 6$.\\
			$\beta_2:$ $U_1\leq 8\;\wedge U_2\leq 7\;\wedge X\geq 6$. & $\beta_5:$ $U_1'\leq 8\;\wedge U_2'\leq 7\;\wedge U_1\geq 6$.\\
			$\beta_3:$ $V_1\leq 8\;\wedge V_2\leq 7\;\wedge Y\geq 6$. & $\beta_6:$ $V_1'\leq 8\;\wedge V_2'\leq 7\;\wedge V_1\geq 6$.\\
		\end{tabular}
	\end{center}
	
%
	We now refer to Figure~\ref{fig:intuition-bodyq2} for visualization regarding Proposition \ref{trick-pro} using the above queries. The circles in the figure represent the mappings $\mu_1, \dots, \mu_6$, and the dots are the variables of $Q_2$. Notice now the intersections between the circles.
	Proposition \ref{trick-pro} refers to these intersections, such as the one between $\mu_3$ and $\mu_6$ (or,  the one between $\mu_2$ and $\mu_5$).
	
	The AC $V_1\geq 6$ ($V_1$ is included in the intersection between $\mu_3$ and $\mu_6$) is the one that is not implied  by $\beta$, as stated in the case (ii) of the Proposition \ref{trick-pro}. In particular, it is easy to verify that the following are true:
	
	\squishlist
		\item $\beta\wedge \neg (V_1\geq 6)\Rightarrow\beta_1\vee\beta_2\vee\beta_3\vee\beta_4\vee\beta_5$.
		\item $\beta\Rightarrow\beta_6\vee \neg (V_1\geq 6)$ (i.e., $\beta\Rightarrow(V_1'\leq 8\;\wedge V_2'\leq 7\;\wedge V_1\geq 6)\vee \neg (V_1\geq 6)$).
		\item $\beta\Rightarrow(V_1'\leq 8)$ and $\beta\Rightarrow(V_2'\leq 7)$.
	\squishend
\end{example}

%
%
%
%

%
%

\section{Proof of Theorem \ref{thm:main123}}
\label{prf:thm-main123p}
\begin{proof}
	We consider the canonical database, $D$, of $Q_2^{CQ}.$ For convenience, the constants in the canonical database use the lower case letters of the variables they represent. Thus, constant $x$ is used in the canonical database to represent the variable $X$. We will use the containment test according to which a
	Datalog query contains a conjunctive query $Q$  if and only if the Datalog query computes the head of $Q$  when applied on the canonical database of the conjunctive query $Q$. All the containment mappings we use in this proof are from the relational subgoals of CQAC $Q_1$ to the relational subgoals of CQAC $Q_2$. Hence, we will refer to them simply as ``mappings.''

\squishlist
\item
We refer to the fact 
 $I_{\theta
		c}(x)$ (atom  $J_{\theta c}(x)$, respectively) as the {\em associated} fact of SI AC $X\theta c$ and vice versa.
\item We will refer to an $I$ fact or  a $J$ fact  when we do not have any use of specifying which specific fact or what is the subscript of $I$ or $J$. When we want to refer to a specific constant, $c$, but not in the subscript, we will use the notation $I(c)$ or $J(c)$.
\squishend
	
	%
	%
	


We assume: 	 $\beta _1= \bigwedge_{i=1\ldots M} e^{\beta_1}_i$, where $M$ is the number of ACs in $Q_1$.

	
	{\bf ``If'' direction.}
%
%
%
%
	{\bf Inductive Hypothesis:}
Let $k<n$. Suppose 
  a $J$ fact associated with AC $e$, is computed using the
	mappings $\mu_1, \mu_2, \ldots, \mu_k$ (via applications of mapping rules). Then the following is true:
	\begin{equation}\label{imply-yy-A}
\beta_2\Rightarrow \mu_1(\beta_1) \vee \mu_2(\beta_1) \vee \cdots \vee \mu_k(\beta_1)\vee \neg \mu_{h}(e^{\beta_1}_i)
\end{equation}
where $\mu_h$ is one of the $\mu_1,\ldots, \mu_k$ and $e=  \mu_h(e^{\beta_1}_i)$.
	
	{\sl Proof of Inductive Hypothesis: }
Suppose the inductive hypothesis is true for $k< n$. We will prove it for $k=n$.
For the base case, suppose we compute a $J$ fact associated with AC $\mu(e^{\beta_1}_i)$ by using only one mapping rule, $\mu$. Hence, we use only link rules
to compute the $I$ facts in the body of the mapping rule. It is easy to see that the following is true:
$\beta_2\Rightarrow \mu(\beta_1) \vee \neg \mu(e^{\beta_1}_i).$

Now we prove the inductive case.
We compute a new $J$ fact firing a mapping rule using mapping $\mu$ and the already computed  $I$ facts, $I(x_1),\ldots, I(x_{m-1})$.  Each such I fact
is computed in turn  via a J fact using either a link rule or a coupling rule. 
%

Suppose $I(x_i)$ is associated with AC $ \mu(e^{\beta_1}_{i}) $ and is  computed either via a link rule or via a coupling rule which used the already computed $J$ fact $J(x_{k_i})$. Suppose     fact $J(x_{k_i})$ is computed via a mapping rule  using mapping $\mu_{m_i}$ and hence $J(x_{k_i})$ is associated with an
AC $\mu_{m_i}(e^{\beta_1}_{l_i})$. 

From
Lemma 
\ref{lem-coupling-implication} we have:
  $\beta_2 \Rightarrow \mu(e^{\beta_1}_{i}) \vee  \mu_{m_i}(e^{\beta_1}_{l_i})$
or equivalently 
$\beta_2 \wedge  \neg\mu(e^{\beta_1}_{i}) \Rightarrow   \mu_{m_i}(e^{\beta_1}_{l_i})$
or equivalently 
\begin{equation}
\label{implyy-C2}   
\beta_2 \wedge   
\neg  \mu_{m_i}(e^{\beta_1}_{l_i}) \Rightarrow   \mu(e^{\beta_1}_{i})
\end{equation}

%
%
%


Inductively, implication \ref{imply-yy-A} for to the already computed fact $J(x_{k_i})$ can be equivalently rewritten as 
	\begin{equation}\label{imply-yy-A1}
\beta_2 \bigwedge_{\mu_j \in Si}\neg \mu_j(\beta_1) \Rightarrow \neg  \mu_{m_i}(e^{\beta_1}_{l_i}) 
\end{equation}

Summing up over all $J$ facts that produce the $I $ facts that participate in firing a mapping rule via mapping $\mu$ to produce the new $J $ fact  we have:

%
%
%


$$\beta_2 \bigwedge_{\mu_j\in S_1\cup S_2\cdots \cup S_{m-1}} \neg  \mu_j(\beta_1) \Rightarrow \bigwedge_{all~i}
 \neg \mu_{m_i}(e^{\beta_1}_{l_i}))$$

From implication \ref{implyy-C2}    
and the above 

$$\beta_2 \bigwedge_{\mu_j\in S_1\cup S_2\cdots \cup S_{m-1}} \neg  \mu_j(\beta_1) \Rightarrow \bigwedge_{all~i}
 \neg \mu_{m_i}(e^{\beta_1}_{l_i}) ))\wedge \beta_2$$
$$ \Rightarrow \bigwedge_{all~i ~except~ i=1}
  \mu (e^{\beta_1}_{i})=\mu(\beta_1) \vee \neg\mu (e^{\beta_1}_{1})$$

The above is equivalent to 

$$\beta_2 \Rightarrow \bigvee_{\mu_j\in S_1\cup S_2\cdots \cup S_{m-1}}  \mu_j(\beta_1)  ~\vee ~
\mu(\beta_1) ~\vee ~\neg\mu (e^{\beta_1}_{1})$$

which proves the inductive hypothesis.

	{\bf  ``Only-if'' direction}.
	{\bf Inductive hypothesis }
Let $l\leq n$. 
There are mappings, $\mu_i, i=1,\ldots m+l$\footnote{We assume $\{    \mu_1, \dots , \mu_m \}\cap \{ \mu _{m+1}, \ldots,   \mu _{m+l} \}=\emptyset$},  such that  the following is true\footnote{Without loss of generality, we use the notation $e^{\beta_1}_{m+j}$, instead of the more precise $e^{\beta_1}_{f_{m+j}}$, since the specific value  plays no role in the proof, only how ACs ``interact'' with each other.}:
\begin{equation}
\label{eq-A}
\beta_2  \Rightarrow \mu_1(\beta_1) \vee \cdots \vee \mu_m(\beta_1) \vee 
\neg\mu _{m+1}(e^{\beta_1}_{m+1}) \vee \neg\mu _{m+2}(e^{\beta_1}_{m+2}) \vee  \cdots  \vee  \neg\mu _{m+l}(e^{\beta_1}_{m+l}) 
\end{equation}

%

where $\mu _{m+i}(e^{\beta_1}_{m+i}), i=1,\ldots,l$ is an AC that uses a variable  corresponding to the canonical database constant $x_{m+i}$ such that $J(x_{m+i})$ has been computed by firing a mapping rule using $\mu _{m+i}$.

	\textit{Proof of the inductive hypothesis:}
We will prove the inductive hypothesis for $l=n$.
For the base case, the following  is true 
	$\beta_2  \Rightarrow \mu_1(\beta_1) \vee \cdots \vee \mu_l(\beta_1)$,  and there are no facts computed yet.

%

%

For the inductive step,  when we consider any  implication   as in \ref{eq-A}.
 We rewrite it
 by transfering 
$\mu _{m+1}(e^{\beta_1}_{m+1}) \ldots  \mu _{m+l}(e^{\beta_1}_{m+l}) $
 on the lhs. Thus, we have again a 
containment entailment on which 
 we use Proposition \ref{trick-pro} to argue that, 
there is a mapping, say $\mu_1$,  among the mappings $\mu_1, \ldots, \mu_m$, for which all the conjuncts are implied by the lhs of the rewritten implication \ref{eq-A}  except one, say $\mu _1(e^{\beta_1}_{1})$.  

Thus, if we consider  $\mu _1(e^{\beta_1}_{i})$ for all $i\neq 1$, the following is true:
\begin{equation}
\label{eq-zz}
\beta_2 \bigwedge_i   \mu _{m+i}(e^{\beta_1}_{m+i}) \Rightarrow  \mu _{1}(e^{\beta_1}_{i}) 
\end{equation}
 or equivalently
\begin{equation}
\label{eq-zz-above}
\beta_2 \Rightarrow \bigvee_i   \neg\mu _{m+i}(e^{\beta_1}_{m+i}) \vee \mu _{1}(e^{\beta_1}_{i}) 
\end{equation}

We consider implication \ref{eq-zz}
and sum up for all  $i$ in $\mu _{1}(e^{\beta_1}_{i})$:

 \begin{equation}
\label{eq-zz-total}
\beta_2 \bigwedge_i   \mu _{m+i}(e^{\beta_1}_{m+i}) \Rightarrow  \bigwedge_{i\neq 1} \mu _{1}(e^{\beta_1}_{i}) 
\end{equation}
We observe that $ \bigwedge_{i\neq 1} \mu _{1}(e^{\beta_1}_{i})$ is equivalent to $\mu_1(\beta_1)  \vee \neg\mu _1(e^{\beta_1}_{i})$.  Thus, the following is true:
\begin{equation}
\label{eq-B.4}
\beta_2  \Rightarrow \mu_2(\beta_1) \vee \cdots \vee \mu_m(\beta_1) \vee \mu _{1}(e^{\beta_1}_{1})  \vee \\
	\neg\mu _{m+1}(e^{\beta_1}_{m+1}) \vee \neg\mu _{m+2}(e^{\beta_1}_{m+2}) \vee  \cdots  \vee \neg\mu _{m+l}(e^{\beta_1}_{m+l})
\end{equation}
 which proves  implication    \ref{eq-A}  for $l=n+1$.
Now we need to prove that some coupling rules can be fired to compute $I$ facts and then a mapping rule can be fired using mapping $\mu_1$ to compute a new $J$ fact on the constant $x_{new}$ represented by the variable 
 $\mu _1(e^{\beta_1}_{1})$, and, thus, conclude the proof of the inductive hypothesis.

We use  Proposition \ref{pro-lsi-rsi-any-closed}, to argue that, in  implication \ref{eq-zz-above}, 
when in  minimal form
there are at most two SIs on the rhs.

We denote $\delta =\bigwedge_i \neg  \mu _{m+i}(e^{\beta_1}_{m+i}) $ and rewrite  implication \ref{eq-zz}
as $\beta_2 \wedge \delta \Rightarrow  \mu _{1}(e^{\beta_1}_{i}) $ for all $i\neq 1$. 
%


According to  Lemma \ref{lemm-3clauses}(4), when we prove a closed AC  is in the closure of $F$ using a chain, then there is at most one constant in this chain. Since $ \mu _{1}(e^{\beta_1}_{i})$ is an SI, this constant is in the end of the chain, hence at most one SI appears in this chain, call this SI $ a_{special}  $.   $ a_{special}  $ is either one of the conjuncts of $\beta_2$ or one of $\delta$. In the former case, we have that $\beta_2\Rightarrow $   $a_{special}. $ In the latter case, we have  that  $\beta_2  \Rightarrow \mu _{1}(e^{\beta_1}_{i}) \vee  \mu _{j}(e^{\beta_1}_{j}) $   where  $a_{special}=    \mu _{j}(e^{\beta_1}_{j}) $ and   $\mu_j$ is one of $\mu _{m+1}, \ldots,   \mu _{m+l}$. 

Now, the already computed J fact associated with 
$ \mu _{j}(e^{\beta_1}_{j}) $  can be used to fire either a link rule or a coupling rule and compute the I associated with $\mu _{1}(e^{\beta_1}_{i})$. Now we have the $m-1$ $I $ facts to compute the new J fact which is associated with $\mu _{1}(e^{\beta_1}_{1})$.

In the end, we need to observe that when the query rule is fired then, the containment entailment can be proven  taking into account the  the inductive hypothesis in the ``if'' direction.
For the other direction, we want to prove  that: if the containment entailment  is true, then there is a computation. We use the inductive hypothesis of the ``only if'' direction. 
When in implication \ref{eq-A} there is one  mapping, $\mu_0(\beta_1)$, then we apply the query rule and compute the head of the query.
	\end{proof}
	%
	%
	
	%
	%
	%


The following is an example on which we can argue that the EDB predicate $U$ is needed.  We leave the analysis to the reader.
\begin{center}
		\begin{tabular}{lll}
			$Q_1({})$ & $\symif$ &$ e(X,Y),e(Y,Z),X\geq 6,Z\leq 7.$ \\
		\end{tabular}
	\begin{tabular}{ll}
		$Q_2()\symif$ & $e(A,B),e(B,C),e(C',D),e(D,E), C'\geq C, A\geq 6, E\leq 7.$ \\
	\end{tabular}
\end{center}
This is the containment entailment:
$ C'\geq C \wedge A\geq 6  \wedge E\leq 7 \Rightarrow (A\geq 6\wedge C\leq 7) \vee (C'\geq 6\wedge E\leq 7). $

\section{Building MCR for RSI1+ query}
\label{sec-extend-single-mapp}
Now, we exploit the results of Section \ref{sec-exte-single-mappinga-var}. We consider queries that are allowed to have any ACs among  head variables.
For ease of reference, we define a CQAC  RSI1+ query to be a CQAC query that
the ACs on  nondistinguished variables are closed semi-interval ACs  and there is a single right semi-interval AC.

Now, we present the algorithm for building an MCR in the language of (possibly infinite) union of CQACs for the case of CQAC views and queries that are RSI1+. The algorithm for building  an MCR  for query $Q$ and viewset $\V$ is the following:

{\bf Algorithm MCR-RSI1+}:

\squishlist
	\item [1.] We consider the query $Q'$ which results from the given query $Q$ after we have removed the ACs that contain only head variables.

\item [2.] We apply the algorithm for building MCR for query $Q'$ and views $\V$ (from previous subsection). Let this MCR be
$R'_{MCR}$.

\item [3.] We add a new rule in $R'_{MCR}$ (and obtain $R_{MCR}$ ) to compute the query predicate $Q$ as follows:
	$$Q():- Q'(), ac_1, ac_2, \ldots$$
	where $ac_1, ac_2, \ldots$ are the ACs that we removed in the first step of the present algorithm.
\squishend

\subsection{Proof that the algorithm {\bf MCR-RSI1+} is correct}
We consider the found by the  {\bf Algorithm MCR-RSI1+}  Datalog$^{AC}$ program, $R_{MCR}$.
Theorem	\ref{thm-mcr-cont1} below proves that every CQAC contained rewriting is contained in $R_{MCR}$ and Theorem
\ref{thm-mcr-in-query1} proves that $R_{MCR}$ is a contained rewriting.

%
%
%


\begin{theorem}
	\label{thm-mcr-cont1}
Given a query $Q$ which is RSI1+ and views $\V$ which are
	CQACs, the following is true:
	Let $R$ be a  CQAC contained rewriting to $Q$ in terms of $\V$. Then $R$ is contained in the one found by the  {\bf Algorithm MCR-RSI1+}  Datalog$^{AC}$ program, $R_{MCR}$.
\end{theorem}
\begin{proof}
Let $R$ be a contained rewriting to query $Q$. Since $Q'$ contains $Q$, $R$ is a contained rewriting of $Q'$ too. Hence, according to the results Theorem \ref{thm-mcr-cont}, $R$ is contained to $R'_{MCR}$.

Since $R$ is contained to $Q$, we consider the view-expansion of $R$, let it be $R_{exp}$ and we know that this is contained in $Q$, hence the containment entailment is true. However, $Q$ is a RSI1+ query, hence we can, according to Section \ref{sec-exte-single-mappinga-var} break the containment entailment in two as follows:

%

	\begin{center}
	\begin{tabular}{ll}
		$\beta_2\Rightarrow \mu_1(\beta_{Q'}) \vee \cdots $&{\em  reduced containment entailment}\\
		$\beta_2\Rightarrow \mu_1(\beta_{Q-head}) $ & eq. (1)\\
	\end{tabular}
\end{center}

where $\beta_2 $ is the conjunction of ACs in the closure of ACs in $R_{exp} $ and $\beta_{Q'}$ is the conjunction of ACs in $Q'$, $\beta_{Q-head}$  is the conjunction of ACs that use only head variables, and $\mu_i$'s are all the mappings from $Q$ to $R_{exp} $.

Observe that in equation (1), we can replace $\beta_2$ with only those ACs in the closure of $\beta_2$ that involve head variables. Because $R$ is AC-rectified, all these ACs appear in $R$; let us denote them by $\beta_{head}$
Thus $\beta_{head}$ logically implies $ \beta_{Q-head}$.
Now, $R'_{MCR}$ and $R_{MCR}$ have the same Datalog-expansions, except that the latter has the ACs in $ \beta_{Q-head}$ as well. 

Hence we have concluded that a) $R$ is contained to $R'_{MCR}$ and b) the ACs in $R$ imply the added ACs in each Datalog-expansion of $R'_{MCR}$ to make a Datalog-expansion of $R_{MCR}$. Hence $R$ is contained in $R_{MCR}$ according to the results in Section \ref{sec-exte-single-mappinga-var}.
%
\end{proof}

\begin{theorem}
\label{thm-mcr-in-query1}
Given a query $Q$ which is CQAC-RSI1 and views $\V$ which are
	CQACs, the following is true:
	The found by the algorithm in Subsection~\ref{subsec-buildingMCRs} Datalog$^{AC}$ program, $R_{MCR}$,   is a contained rewriting.
\end{theorem}
\begin{proof}


Let $E$ be a CQAC query which is a Datalog-expansion of $R_{MCR}$. Let $E'$ be the CQAC that results from
$E$ by removing the head ACs.  $E'$ is a contained rewriting in $Q'$.
Hence if we consider the view-expansion, $E'_{exp}$, of $E'$, the containment entailment is true for $E'_{exp}$ and $Q'$.

Moreover, trivially we have $\beta_{E}\Rightarrow \beta_{Q-head}$
and using the distributive law, we derive the containment entailment that shows containment of the view-expansion
$E_{exp}$ of $E$ to $Q$.
%
\end{proof}


A straightforward consequence of the above two theorems is the following theorem which is the main result of this section.
\begin{theorem}
	\label{thm-mainsec6}
	Given a query $Q$ which is CQAC-SI1+ and views $\V$ which are
	CQACs, the algorithm in Subsection~\ref{subsec-buildingMCRs} finds an MCR of $Q$ using $\V$ in the language of (possibly infinite) union of CQACs which is expressed by a Datalog$^{AC}$ query.  Hence, for this case of query and views, we can compute certain answers in polynomial time.
\end{theorem}


\end{document}